\documentclass{article}
\usepackage{amsmath,amsfonts,amssymb,amsthm}
\usepackage{kbordermatrix}
\usepackage{graphicx}
\usepackage{booktabs}
\usepackage[round]{natbib}
\usepackage{subfig}
\usepackage{microtype}

\title{Multibrand geographic experiments}
\date{October 2016}
\author{Art B. Owen\\Google Inc.\and Tristan Launay\\Google Inc.}

\newtheorem{theorem}{Theorem}

\theoremstyle{definition}
\newtheorem{definition}{Definition}

\newcommand{\post}{\mathrm{post}}
\newcommand{\pre}{\mathrm{pre}}
\newcommand{\err}{\varepsilon}
\newcommand{\dnorm}{\mathcal{N}}
\newcommand{\hbbeta}{\hat{\bar\beta}}
\newcommand{\dgam}{\mathrm{Gam}}
\newcommand{\var}{\mathrm{var}}
\newcommand{\wh}{\widehat}

\newcommand{\dustd}{\mathbf{U}}

\newcommand{\cv}{\mathrm{cv}}

\newcommand{\simind}{\stackrel{\mathrm{ind}}\sim}
\newcommand{\simiid}{\stackrel{\mathrm{iid}}\sim}
\newcommand{\tran}{\mathsf{T}}

\newcommand{\phm}{\phantom{-}}

\newcommand{\e}{\mathbb{E}}

\begin{document}
\maketitle

\begin{abstract}
In a geographic experiment to measure advertising effectiveness,
some regions (hereafter GEOs) get increased advertising while others do not.
This paper looks at running $B>1$ such experiments simultaneously
on $B$ different brands in $G$ GEOs, and then using shrinkage methods to estimate
returns to advertising.  There are important practical gains from
doing this.  Data from any one brand helps to estimate the return
of all other brands. We see this in both a frequentist and Bayesian
formulation. As a result, each individual experiment could be made
smaller and less expensive when they are analyzed together.
We also provide an experimental design for multibrand experiments
where half of the brands have increased spend in each GEO while
half of the GEOs have increased spend for each brand.  
For $G>B$ the design is a two level factorial for each brand
and simultaneously a supersaturated design for the GEOs.
Multiple simultaneous experiments also allow one to 
identify GEOs in which advertising is generally more effective.
That cannot be done in the single brand experiments we consider.
\end{abstract}

\section{Introduction}

It is difficult to measure the impact of advertising
even in the online setting where responses of individual
users can be linked to conversion activities such as
visiting a website or buying a product. Regression models
are often fit to such rich observational data. While insights
from observational data are suggestive, they 
seldom establish causal relations.

Google has expertise in using geographical experiments
to measure the causal impact of increased advertising,
as decribed by \cite{vave:koeh:2011,vave:koeh:2012}.
Advertising is increased in some regions and left constant
or decreased in others (the control regions). 
Then the corresponding values of some
key performance indicator (KPI) are measured and related
to the spending level. We will call
the regions GEOs.  The Nielsen company has designated
market areas (DMAs) and television market areas (TMAs).
GEOs are similar but not necessarily identical to these.

Other things being equal, it is easier to measure the
impact of a large advertising change than a small one.
Having two widely separated spend levels makes for a 
more informative experimental design. There are however practical and
organizational constraints on the size of an experimental
intervention. Advertising managers may be reluctant to experiment with
large spend changes.  Also, in a small GEO, there may not be
enough inventory of ad impressions to sustain a large spending
increase.

Both of these problems can be mitigated by experimenting
on several brands at once. The experimental design is
like the one sketched below.
\begin{align*}
\begin{matrix}
& \text{GEO 1} & \text{GEO 2} & \text{GEO 3} & \text{GEO 4} & \cdots & \text{GEO G} \\
\text{Brand 1\,} & + & - & + & - &\cdots& +\\
\text{Brand 2\,} & - & - & + & + &\cdots& -\\
\vdots &\vdots &\vdots &\vdots &\vdots &\ddots &\vdots\\
\text{Brand B} & - & + & - & + &\cdots &+
\end{matrix}
\end{align*}
Here the experiment gives Brand 1 an increased spending
level in GEOs 1 and 3 and the control level of
spending in GEOs 2 and 4. Every brand gets increased
spend in half of the GEOs, with each GEO being in the
test group for some brands and the control group for others. 
The combined information
from all B brands can then be used to get a good
measure of the overall effectiveness of advertising.
Using shrinkage methods it is also possible for the
data from one brand to improve estimation for another one.
Because the multibrand experiment pools information, it
can be run with smaller spending changes than we would
need in single brand experiments.

An outline of this note is as follows.
Section~\ref{sec:models} presents regression models
for single brands and multiple brands.
Section~\ref{sec:design} gives a scrambled checkerboard 
experimental design
in which half of the GEOs are treatment for each brand
and half of the brands get the treatment level in each GEO.
Subject to these constraints, there may be weak
correlations among pairs of brands or among pairs of GEOs.
Section~\ref{sec:design} also shows that
certain classical designs (balanced incomplete blocks
and Hadamard matrices) that might seem appropriate
are, in fact, not well suited to this problem.
Section~\ref{sec:regression} simulates a single brand
experiment $1000$ times over $20$ GEOs. 
The true return to advertising in those simulations is $\beta=5$.
There is reasonable power to detect $\beta\ne0$ when 
advertising is increased by $1$\% of prior period sales,
but not when it is increased by only $0.5$\% of prior sales.
In either case the standard error of the estimated return
is quite large.
Section~\ref{sec:multibrandsimu} describes a multibrand simulation
with $30$ brands in $20$ GEOs. The advertising return
for brand $b$ is $\beta_b\sim\dnorm(5,1)$. The
estimator of \cite{xie:kou:brow:2012} that shrinks each brand's
parameter estimate towards their common average is about $3.2$
times as efficient at estimating $\beta_b$ than using only
that brand's data, when the treatment is $1$\% of sales.
For smaller treatments, $0.5$\% of sales, shrinkage is about $7.8$
times as efficient as single brand experiments.
Some simulation details are placed
in Section~\ref{sec:simulation}.
Section~\ref{sec:fullbayes} simulates a fully Bayesian analysis.
The simulation there has $G=160$ GEOs but only $B=4$ brands 
and it also shows a strong benefit from pooling.
The Bayesian method has similar accuracy to Stein shrinkage and comes with easily
computed posterior credible intervals.
Section~\ref{sec:discussion} has some conclusions and discussion.

\section{Single- and multi-brand models}\label{sec:models}

We target an experiment comparing an $8$ week
background period followed by a $4$ week experimental period. 
To prepare for this project, data from $5$ very
different advertisers was investigated.
The industries represented were:
hair care, cosmetics, outdoor clothing,
photography and baked goods. There were strong
similarities in the data for all of these industries.

If one plots the $8$ week KPI for a brand versus
the prior $4$ week KPI for that brand, using one point
per GEO the resulting points fall very close to a straight
line on a log-log plot, in all $5$ data sets. 
The linear pattern is so strong
because the GEOs vary immensely in size.

Inspecting all of that data it became clear that
the following model was a good description
of a single brand's data
\begin{align}\label{eq:singlemodel}
Y^\post_g = \alpha_0 + \alpha_1Y^\pre_g+ \beta X_g^\post + \err_g^\post,\quad 
g=1,\dots,G.
\end{align}
Here $Y^\post_g$ is the KPI for GEO $g$ in the experimental
period, $Y^\pre_g$ is the corresponding value in the pre-experimental
period and $X_g^\post$ is the amount spent on advertising in GEO $g$
in the post period. The basic linear regression
$\alpha_0 + \alpha_1Y^\pre_g$ is strongly predictive, because
the underlying GEO sizes are very stable and the KPI is roughly
proportional to size. There was not an
appreciable week to week autocorrelation for sales data within GEOs.
What little autocorrelation there was would be greatly diminished
for multi-week aggregates such as an $8$ week prior period
followed by a $4$ week experimental one.

Model~\eqref{eq:singlemodel} is the one used by
\cite{vave:koeh:2011}.
The parameter of greatest interest is $\beta$.
When $X_g^\post$ is the dollar amount spent on advertising,
and the KPI $Y_g^\post$ is the revenue in the experimental
period, then $\beta$ is simply the number of incremental
dollars of revenue per dollar spent on advertising.
The interpretability of $\beta$ as a return to advertising
is the reason why
we work with model~\eqref{eq:singlemodel}. Modeling the logarithm
of the KPI would have some statistical advantages,
but it makes for a less directly interpretable~$\beta$.

In simulations, the value of $X_g^\post$ is 
proportional to $Y_g^\pre$. We take $X_g^\post = \delta Y_g^\pre$
in the treatment group and $X_g^\post=0$ in the control group.
Our default choice is $\delta =0.01$,
representing differential spend equal to one percent of prior sales.
This need not mean setting advertising to $0$ in the
control group.  Here $X_g^\post$ is the level of additional
spending above the historic or pre-planned level for that GEO.
In an experiment that reduced spend in some GEOs
to offset increases in others, $X_g^\post$ would be
negative in some GEOs and positive in others.

In model~\eqref{eq:singlemodel}, it is not reasonable
to suppose that the errors $\err_g^\post$ are independent
and identically distributed.  In all five real data sets
it was clear that the standard deviation of the KPI
is larger for larger GEOs. To a very good approximation,
the standard deviation was proportional to the KPI itself.
When simulating model~\eqref{eq:singlemodel}, Gamma random
variables were used instead of Gaussian ones.  The
standard deviation in a Gamma random variable is proportional
to its mean. See Section~\ref{sec:simulation}.

Now suppose that a single advertiser has 
multiple brands $b=1,\dots,B$.
It then pays to experiment on all $B$ brands at once.
In a multibrand setting we can fit the regression model
\begin{align}\label{eq:multimodel}
  Y^\post_{gb} = \alpha_{0b} + \alpha_{1b}Y^\pre_{gb}+ \beta_b X_{gb}^\post + \err_{gb}^\post,\quad b=1,\dots,B,\ g=1,\dots,G.
\end{align}
The brands should be distinct enough that advertising for one
of them does not affect sales for another.  For instance, two different
diet sodas might be too closely related for this model
to be appropriate.

The overall return to advertising is measured by
$$\bar\beta = \frac1B\sum_{b=1}^B\beta_b.$$
A combined experiment will be very informative about $\bar\beta$.
By using Stein shrinkage, the 
combined experiment can also give more accurate estimates
of individual $\beta_b$ than we would get from just an experiment
on brand $b$.

\subsection{Differential GEO responsiveness}
A multibrand experiment can address some issues that 
are impossible to address in a single brand experiment.
Suppose for instance that advertising is more effective
in some GEOs than it is in others. In a single experiment
an unusually responsive or unresponsive GEO might generate
an outlier, but we would not know the reason. From a multibrand
experiment we can fit the model
\begin{align}\label{eq:multimodellocal}
  Y^\post_{gb} = \alpha_{0b} + \alpha_{1b}Y^\pre_{gb}+ (\beta_b+\gamma_g) X_{gb}^\post + \err_{gb}^\post.
\end{align}
The new parameter $\gamma_g$ measures the extent to which
advertising is especially effective  in GEO $g$. In a single
brand experiment with $G$ responses we could not estimate
these per-GEO parameters. It would amount to fitting
$3+G$ regression parameters to $G$ responses. 
In a multibrand experiment we get 
$G\times B$ responses and model~\eqref{eq:multimodellocal}
has only $3B+G$ regression parameters.  If one consistently
sees that some GEOs have better responses to ads than others
then it would be reasonable to focus more advertising in those
GEOs. The parameter $\gamma_g$ can still be practically 
important even when it is not large enough to generate outliers.

\section{Scrambled checkerboard designs}\label{sec:design}

For each brand, we should have half of the GEOs
in the control group and half in the treatment
group.  This necessitates an even number $G$ of GEOs
which is not difficult to arrange.
Similarly, with an even number $B$ of brands,
each GEO should be in the treatment group for 
half of the brands and in the control group for the other half.
We would want to avoid
a situation where a large GEO like Los Angeles was the control 
group for most of the brands, or in the treatment group for
most of the brands.

A second order concern is that we would not want any pair
of brands to always be treated together or in the control
group together. For two brands the four possibilities
$\{ TT, TC, CT,CC\}$ describe GEOs where the first brand
is treatment or control based on the first letter (T or C)
and the second brand's state is given by the second letter.
Ideally we would like all four of these possibilities to arise
equally often for all pairs of brands and an analogous condition
to hold for GEOs.

This second order concern brings to mind balanced incomplete
block (BIB) designs \citep{coch:cox:1957}, but that is a different concept and
a BIB does not actually solve the problem.
See Section~\ref{sec:bib}.
There is also potential for submatrices of Hadamard matrices to
be good designs but that imposes unwanted restrictions on the
numbers $B$ and $G$ of brands and GEOs. See Section~\ref{sec:hadamard}.

\begin{theorem}\label{thm:nocando} 
Suppose that there are $G\ge1$ GEOs and $B\ge1$ brands
where each GEO has the treatment for half of the brands
and each brand is in the treatment group for half of the
GEOs.  Then it is impossible to have all four
combinations $\{ TT, TC, CT,CC\}$ arise equally often for
each distinct pair of GEOs as well as for each distinct
pair of brands.
\end{theorem}
\begin{proof}
If we represent our design by a $G\times B$ matrix $Z$
of $\pm 1$s with $+1$ for treatment and $-1$ for control, then each
row and column of $Z$ must sum to zero. The second order consideration
about pairs $TT$ through $CC$
requires the columns of $Z$ to be orthogonal. Since they
are orthogonal to a column of $1$s there can only be $G-1$
of them at most, so $B\le G-1$. The same argument applied to
rows yields $G\le B-1$
We cannot have both $G<B$ and $B<G$,
so it is impossible to exactly satisfy the second
order conditions. 
\end{proof}

Because the second order considerations cannot possibly
be satisfied, we compromise on them while still
insisting on balance within every row and every column.

A practical approach is to start with a $G\times B$ checkerboard pattern
like that in Figure~\ref{fig:checkerboard}, and randomly perturb it. 
Each brand gets
the treatment in half of the GEOs and conversely each GEO is
in the treatment group for half of the brands. Then we use
random swaps to break up the checkerboard pattern. The second order
criteria are then treated via random balance \citep{satt:1959}.

The swaps are based on a Markov chain studied by
\cite{diac:gang:1995}. Their setup uses $0$s and $1$s
where we have $\pm1$s, but results translate directly
between the two encodings.
We sample two distinct rows and two distinct columns of the grid.
If the pattern in the sampled $2\times2$ submatrix matches
$$
\begin{pmatrix} +&\cdot\\\cdot&+\end{pmatrix}
\quad\text{or}\quad
\begin{pmatrix} \cdot&+\\+&\cdot\end{pmatrix}
$$
then we switch it to the other of these two. 
Here and below we use $\cdot$ in place of $-$ 
where that would improve clarity.
\cite{diac:gang:1995} show that this sampler yields a 
connected symmetric aperiodic Markov chain
on the set of binary $G\times B$ matrices with row
sums equal to $B/2$ and column sums equal to $G/2$.
The stationary distribution is uniform on such matrices.

Their setting was more general: the matrix contained
nonnegative integers with specified row and column sums,
not just $0$s and $1$s.
A verbatim translation of their algorithm would actually
make the proposed switch with probability $1/2$. Raising
the acceptance probability to $1$ for binary matrices
still satisfies detailed balance with respect to the uniform
distribution, so the Markov chain still uniformly samples
the desired set of matrices.

\begin{figure}[t]
\centering
\includegraphics[width=0.99\hsize]{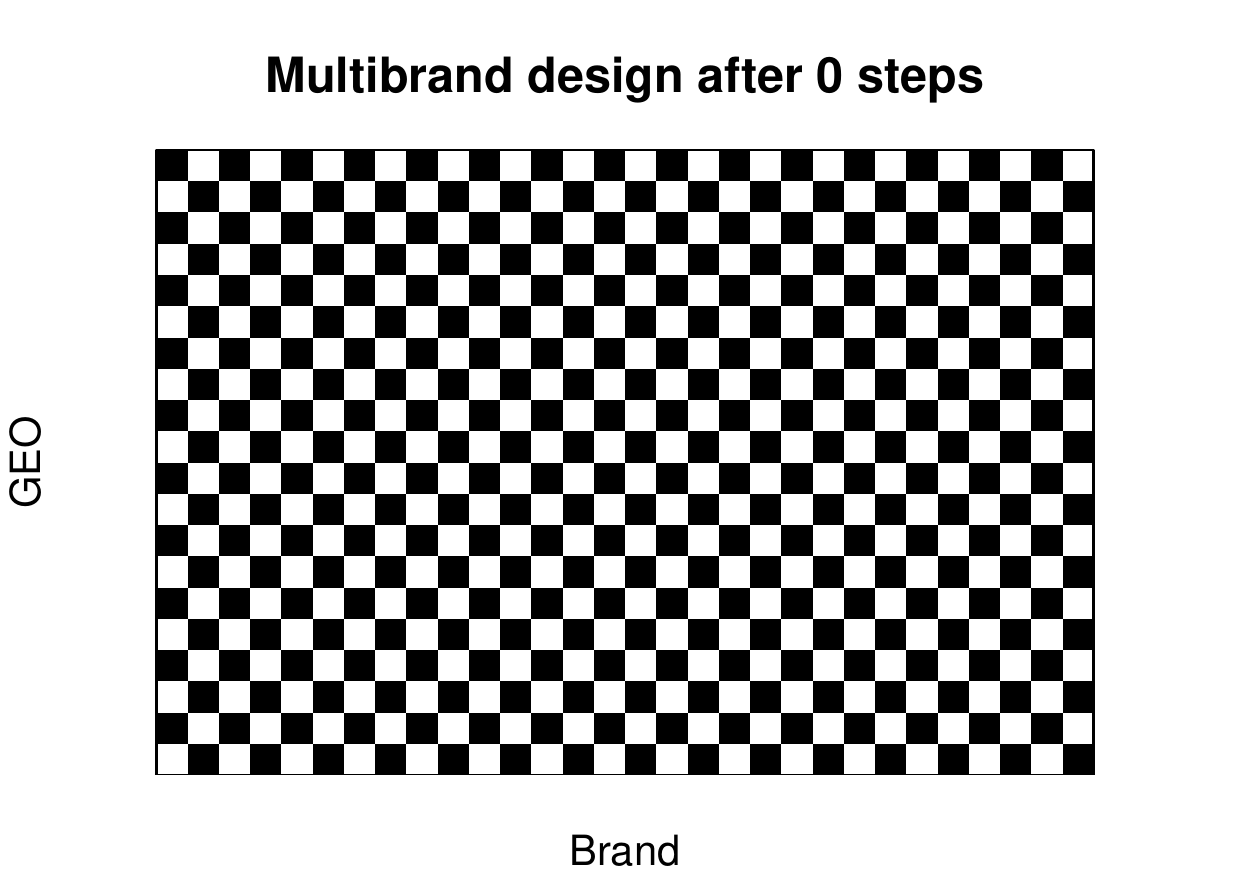}
\caption{\label{fig:checkerboard}
A design where half of the $20$ GEOs are treatment (black)
for each of the $30$ brands and the others are control (white).
Conversely, half of the $30$ brands
are treatment group for each of the $20$ GEOs and half are control. 
This design is unsuitable because any pair of brands either always
get the same allocation or always get an opposite allocation. We
address that problem via scrambling.
}
\end{figure}

Figure~\ref{fig:step100} shows the design after $100$
attempts to flip a $4$-tuple of elements.  The original
checkerboard pattern is still clearly visible and so $100$
attempts are not enough.

There are $16$ possibilities for
any $2\times 2$ submatrix of the design and $2$ of these possibilities
are flippable.  So we should expect that after the algorithm
has been running a while that the chance of a flip is about $1/8$.
The algorithm starts with a $100$\% flippable checkerboard
and so it is reasonable to suppose that the flipping chance starts
above $1/8$ and decreases to that level.
Each flip flips $4$ pixels in the image. Therefore we
reverse about $1/2$ pixels per attempt. 
Figure~\ref{fig:step30000} shows the result after $30{,}000$ attempts
so that the average number of flips per pixel is about $25$.

\begin{figure}[t]
\centering
\includegraphics[width=0.99\hsize]{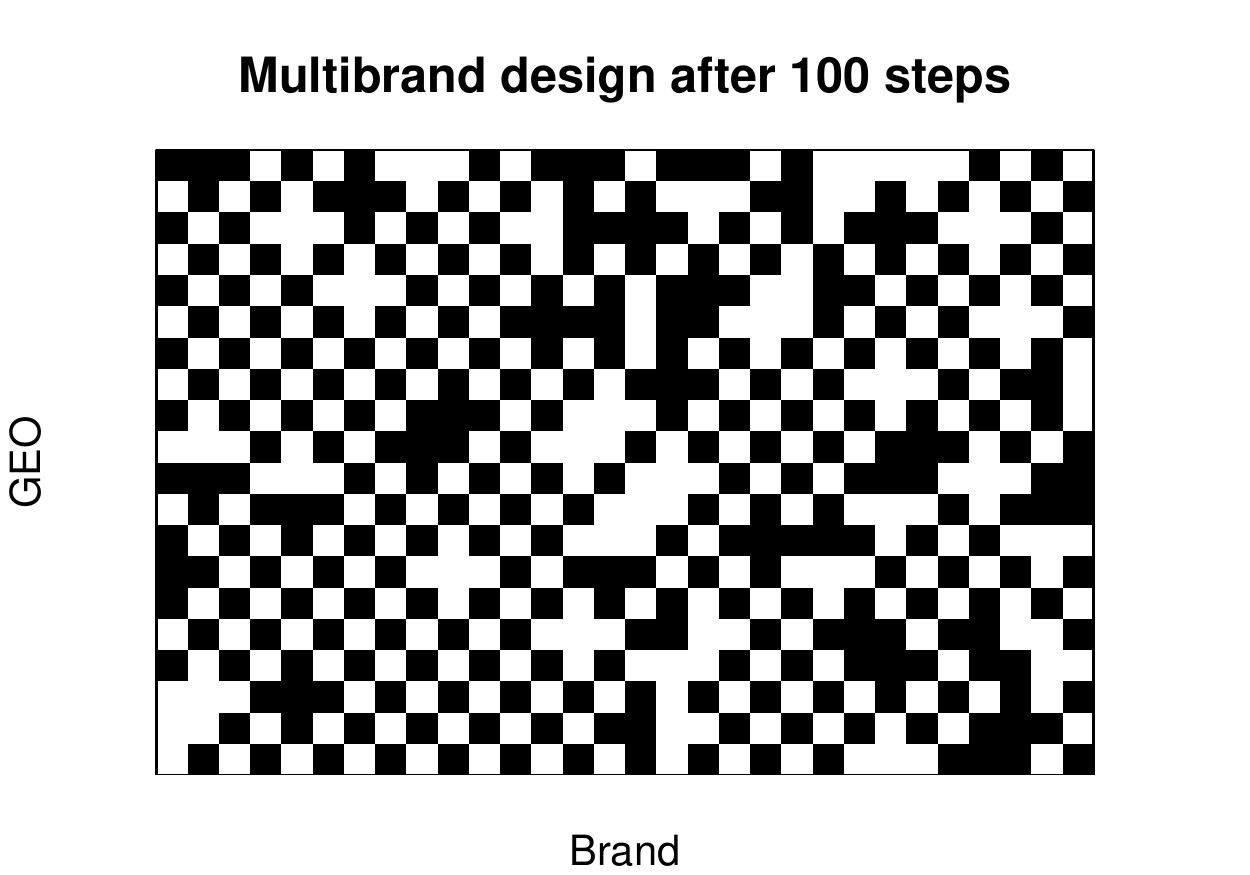}
\caption{\label{fig:step100}
Design from Figure~\ref{fig:checkerboard} after $100$ attempted flips,
showing that more than $100$ attempts are needed.
}
\end{figure}

\begin{figure}[t]
\centering
\includegraphics[width=0.99\hsize]{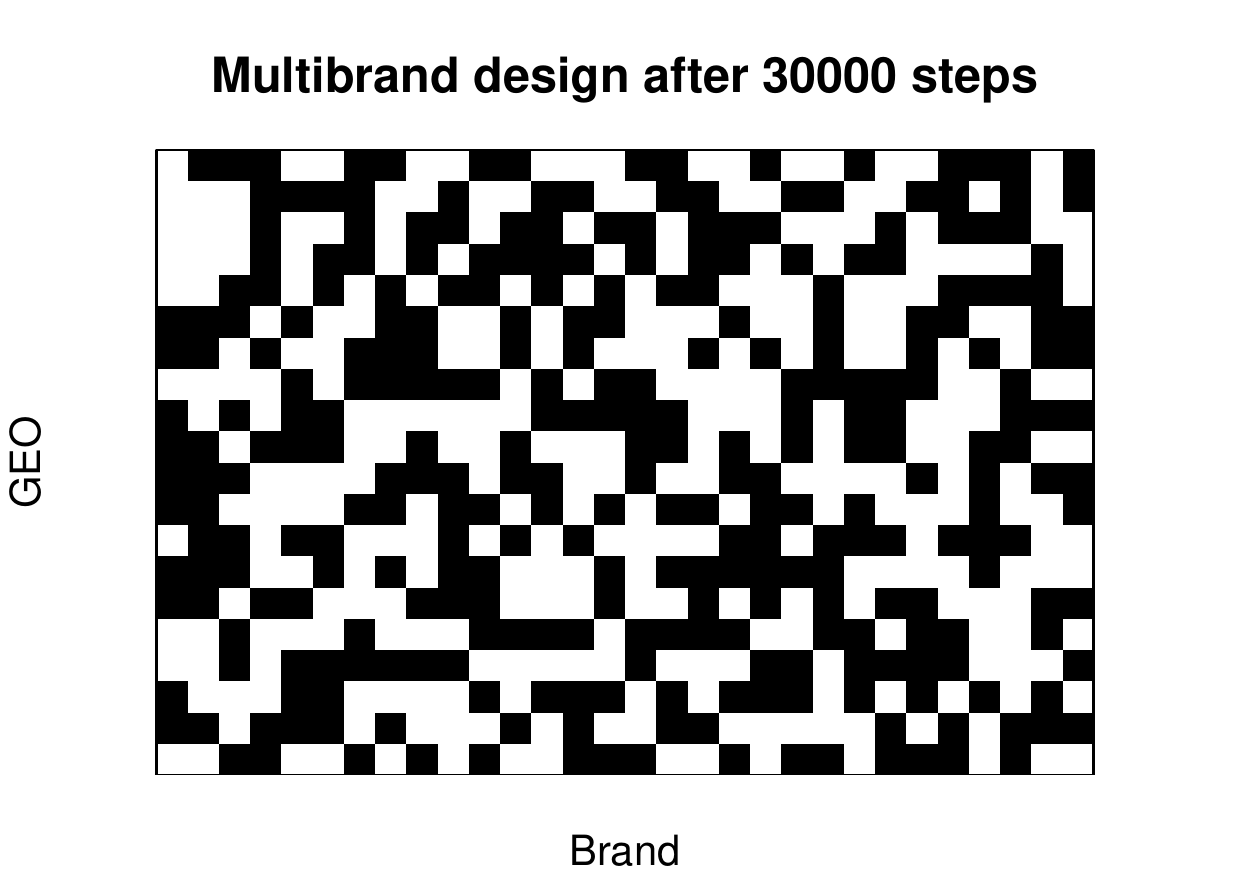}
\caption{\label{fig:step30000}
Design from Figure~\ref{fig:checkerboard} after $30{,}000$ attempted flips.
}
\end{figure}

The algorithm is very fast. To do $90{,}000$ steps on a
larger $60\times30$ grid
takes just over $7$ seconds in R on a commodity PC. It is possible to
do many more flips, but that seems unnecessary.

We can look at the correlations among brands as
the sampling proceeds.  There are $B$ brands and
hence $B(B-1)/2$ different off-diagonal correlations.  The minimum,
maximum and root mean squared correlations among brands
are plotted in Figure~\ref{fig:brandcorr}.  
The same quantities for GEOs are plotted 
in Figure~\ref{fig:geocorr}.
These correlations are remarkably
stable after a short warm-up period. The stability has set in
before $BG/2=300$ successful flips have been made.

\begin{figure}[t]
\centering
\includegraphics[width=0.99\hsize]{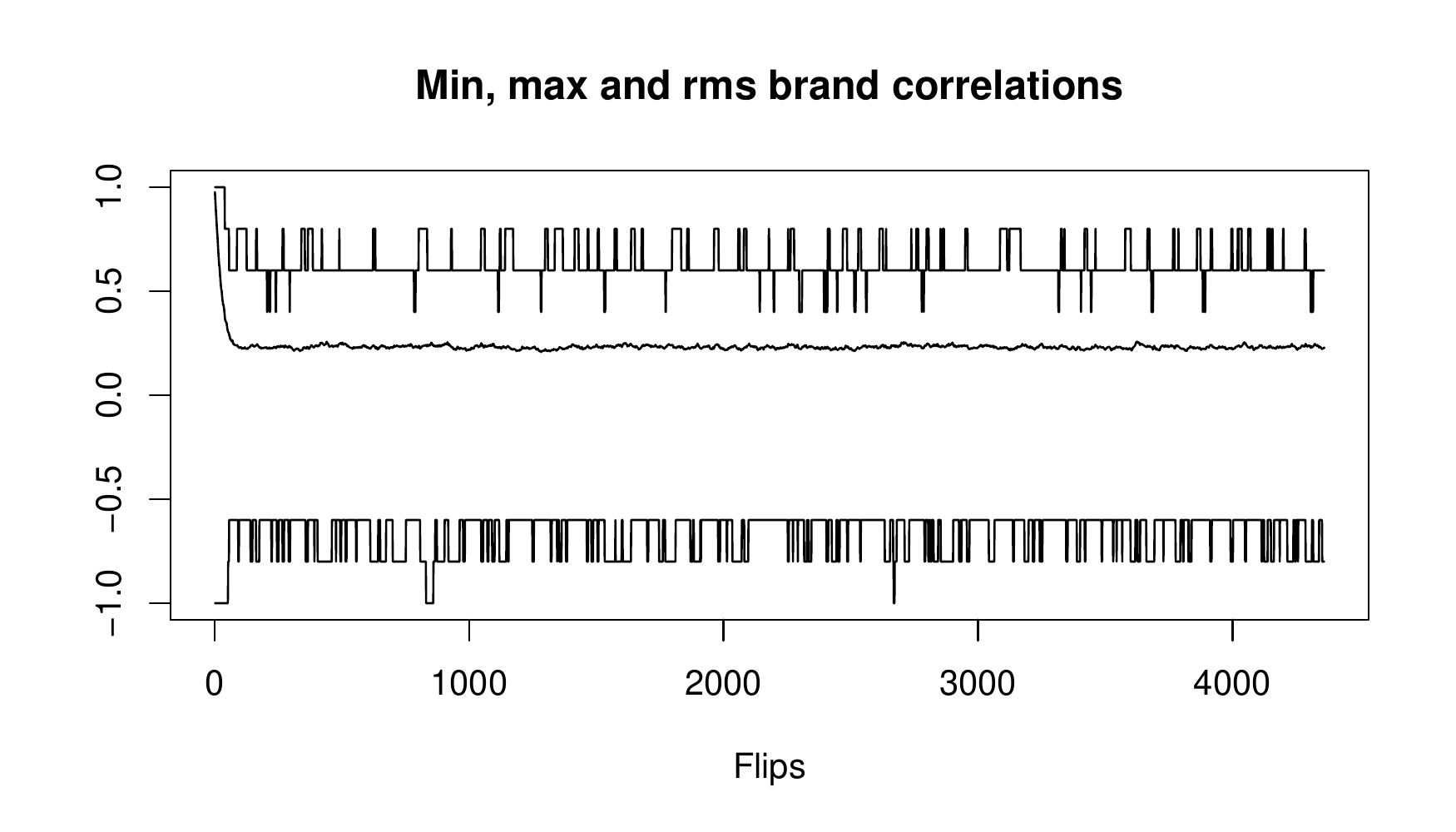}
\caption{\label{fig:brandcorr}
Interbrand correlations as the number of successful flips
increases.
}
\end{figure}

\begin{figure}[t]
\centering
\includegraphics[width=0.99\hsize]{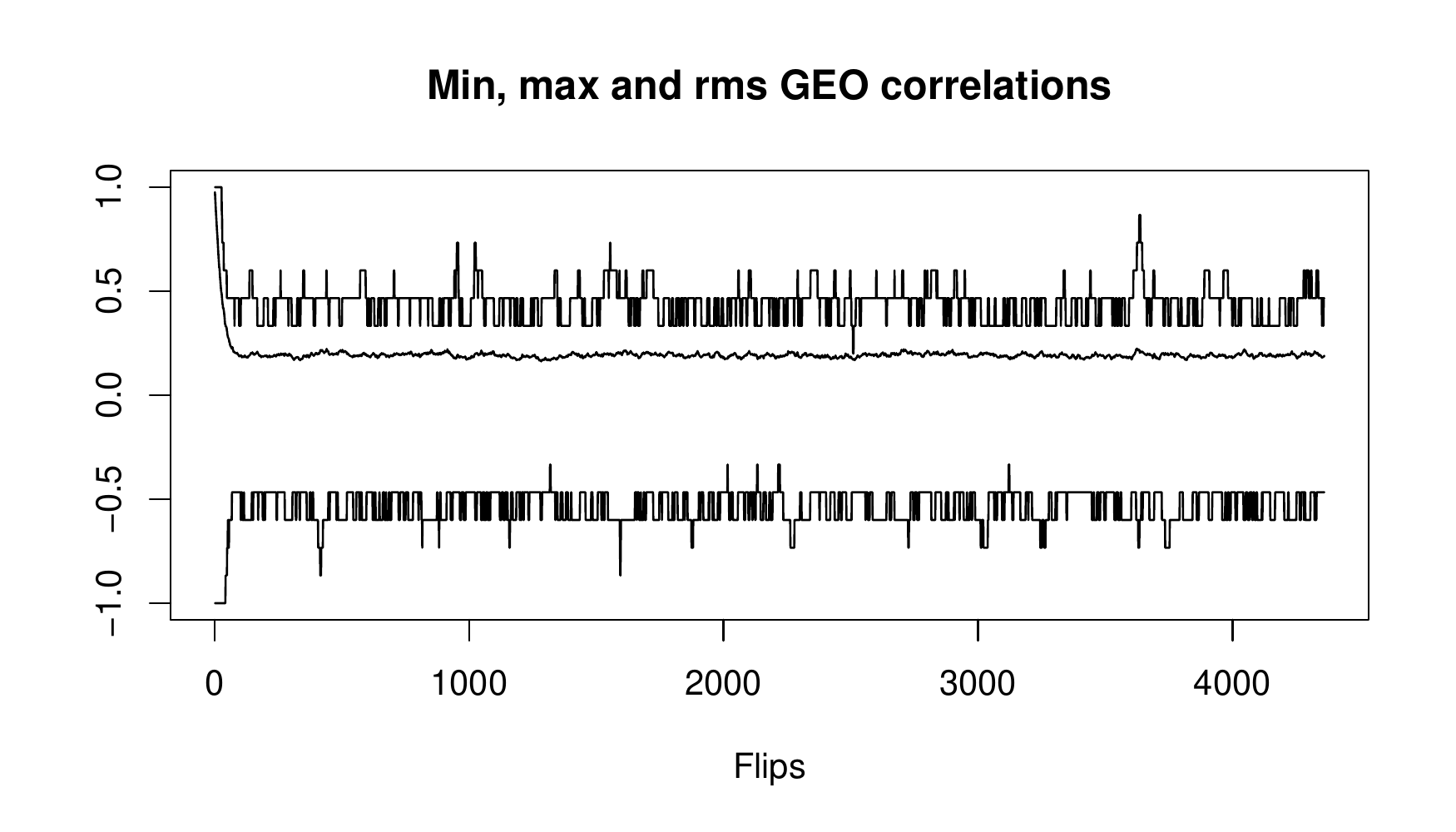}
\caption{\label{fig:geocorr}
InterGEO correlations as the flips proceed.
}
\end{figure}

There is a relationship among the sum of squared GEO
correlations and the sum of squared brand correlations
at every step of the algorithm. 
For brands $b,b'$ their correlation is $\rho_{bb'}=(1/G)\sum_{g=1}^GX_{bg}X_{b'g}$.
For GEOs $g,g'$ their correlation is $\rho_{gg'}=(1/B)\sum_{b=1}^BX_{bg}X_{bg'}$.
Then counting cases $g=g'$ and $b=b'$,
\begin{align*}
\sum_{gg'}\rho_{gg'}^2 &= \frac{S}{B^2},\quad\text{and}\quad 
\sum_{bb'}\rho_{bb'}^2 = \frac{S}{G^2},\quad\text{where}\\
 S & = \sum_g\sum_{g'}\sum_b\sum_{b'}X_{bg}X_{b'g}X_{bg'}X_{b'g'}
\end{align*}
which can be rearranged to get
$$
\sum_{b\ne b'}\rho_{bb'}^2 = \frac{B^2}{G^2}\sum_{g\ne g'}\rho_{gg'}^2+\frac1G.
$$
This phenomenon was noted
by \cite{efro:2008} in some work on doubly standardized matrices
of microarray data.
The mean squared correlation is comparable in size
to what we would get with independent sampling. 
That is, we are able to balance all GEOs and all brands
exactly without paying a high cost on these correlations.

The rest of this section considers classical designs
that do not apply to our situation and then considers
when designs that meet our secondary goals can be constructed.
Some readers might prefer to skip to Section~\ref{sec:regression}
which discusses a simulated example.

\subsection{Designs derived from a BIB}\label{sec:bib}
In a BIB, one compares $B$ quantities in blocks of size $s<B$
and every pair of quantities appears together in the same number
of blocks.
A BIB with block size $s=B/2$ and one block per GEO might
be repurposed for multi-brand experiments 
by making the $B/2$ elements of
each block correspond to brands given the treatment level.
A small example with $B=4$ brands and $G=6$ GEOs looks like
this
$$
\kbordermatrix{&B_1&B_2&B_3&B_4\\
G_1 &+ & +& \cdot & \cdot \\
G_2 &+ & \cdot & + & \cdot \\
G_3 &+ & \cdot & \cdot & +  \\
G_4 & \cdot & + & + & \cdot\\
G_5 & \cdot & + & \cdot & +\\
G_6 & \cdot & \cdot & + & +\\
}
$$
where a $+$ indicates that the given brand gets the
treatment in the given GEO.  The problem is that
GEOs $1$ and $6$ are exact opposites as are GEOs $2$ and $5$
and GEOs $3$ and $4$. Similarly, for any pair of brands 
the matrix
$$
\kbordermatrix{&+&\cdot\\
+ & 1 & 2 \\
\cdot & 2 & 1 \\
}
$$
gives the number of GEOs at each treatment combination. 
We know from Theorem~\ref{thm:nocando} 
that
equal numbers in all four configurations cannot be attained.
Here we see that for this BIB any two brands are more
likely to be at opposite treatment versus control settings
than at the same level.

\subsection{Designs derived from a Hadamard matrix}\label{sec:hadamard}

A Hadamard matrix \citep{heda:sloa:stuf:2012}
$H$ is an $n\times n$ matrix with elements
$\pm 1$ satisfying $H^\tran H=HH^\tran = I_n$.
An example Hadamard matrix with $n=8$ is depicted here:
$$
\kbordermatrix{\\
  & +      & +      & +      & +      & +      & +      & +      & +    \\
  & +    & \cdot  & +    & \cdot  & +    & \cdot  & +    & \cdot\\
  & +      & +    & \cdot& \cdot  & +      & +    & \cdot& \cdot\\
  & +    & \cdot& \cdot  & +      & +    & \cdot& \cdot  & +    \\
  & +      & +      & +      & +    & \cdot& \cdot& \cdot& \cdot\\
  & +    & \cdot  & +    & \cdot& \cdot  & +    & \cdot  & +    \\
  & +      & +    & \cdot& \cdot& \cdot& \cdot  & +      & +    \\
  & +    & \cdot& \cdot  & +    & \cdot  & +      & +    & \cdot\\
}.
$$

Suppose that we use $+$ for treatment  and $\cdot$
for control, and use columns of the design for brands
and rows for GEOs. Column $1$ is not suitable because
it describes a brand that is at the treatment level in
all GEOs.  In  applications, the first column of a Hadamard
matrix corresponds to the intercept term, not one of
the treatment variables, and so we might use the
last $n-1$ columns.

Row $1$ of the matrix above is not suitable as it describes a GEO
that is in the treatment group for all brands.
We can always reverse the sign in $3$ of the $7$ columns
and get a new design.  If we reverse columns 2,3,4 then
row 5 will be all $-1$'s (after the intercept column).
Certain other reversal choices will not produce a degenerate row
but will affect the number of $+1$s in the rows.

Hadamard matrices are potentially useful but require
special conditions. They only exist for $n=1,2$ 
(which are unsuitable) or $n=4m$ for certain
positive integers $m$.  
There are only $12$ integers $m<500$
for which no Hadamard matrix of order $n=4m$
is known.  See \cite{djok:golu:kots:2014}
who shortened that list from $13$ integers
by solving the case $m=251$.

A more serious problem is
that dropping the first column of a Hadamard matrix
and toggling the signs of some columns is only useful
if $B-1=G=4m$ for some $m$. One could drop the first
row too, yielding a design for $B=G=4m-1$ which has
near balance for each brand and each GEO.
But both of these choices impose unwanted restrictions
on $B$ and $G$. In principal one could
take a $G\times B$ submatrix of the last $n-1$
rows and columns of a Hadamard matrix but
then the result is even farther from the desired
balance of having each brand get the control treatment
in $G/2$ GEOs and each GEO delivering the control treatment
to each of $B/2$ brands.

\subsection{Constraints}\label{sec:constraints}

When $B\ge2$ and $G\ge2$ are both even then the
design matrix we want is $B\times G$ binary
matrix with $B/2$ ones in each column and
$G/2$ ones in each row.  Such matrices always
exist.  We would also like, when possible, to
have no two rows or columns be identical, or
to be opposite of each other.  

\begin{definition}
Two vectors $v_1,v_2\in\{-1,1\}^k$ have a collision
if either $v_1=v_2$ or $v_1=-v_2$.
A matrix $X\in\{-1,1\}^{n\times p}$ has no collisions if
no two of its rows have a collision and no two of its
columns has a collision.
\end{definition}

\begin{definition}
  A matrix $X\in\{-1,1\}^{n\times p}$ is balanced if
  each row sums to $0$ and each column sums to $0$.
\end{definition}

Our design uses balanced binary matrices.
Ideally we would like our design matrix to be
free of collisions. This secondary
constraint cannot always be met.
For any even number $B$ there are only
${B\choose B/2}$ different binary
vectors having exactly $B/2$ ones. Because we don't
want duplicates or opposite pairs we must
have $G \le {B\choose B/2}/2$, and conversely
$B \le {G\choose G/2}/2$. 

First, if $B=2$ then any pair of GEOs must get
either the exact same or exact opposite treatment,
and similarly for brands when $G=2$.
So when $\min(B,G)=2$ collisions will occur.

\begin{theorem}
Let $B\ge2$ and $G\ge 2$ be even numbers. If 
$\min(B,G)\le 4$ then there is no balanced binary
$G\times B$ matrix without collisions.
If $B=G=6$ or $B=G=8$, then there is such a matrix.
\end{theorem}
\begin{proof}
If $G=2$ then the result is obvious because the
second row must then be the opposite of the first
one. Similarly if $B=2$, and so no such matrix
is available when $\min(B,G)=2$. 

For $G=B=4$, consider a $4\times4$
matrix of $+$ and $\cdot$ with exactly two $+$
symbols in each row and each column. We can sort
the columns so that the first row is
$\begin{pmatrix} +&+&\cdot&\cdot\end{pmatrix}$.
If the matrix has no collisions, then each subsequent row must have  
exactly one $+$ in the first two columns and one $+$ in the last two columns.
Because each column has two $+$'s, only one of the next three rows can have
a $+$ in column $1$ and only one of those rows can have a $+$ in column $2$.
There is therefore no way to put three more rows into the matrix without
having a collision.  As a result there is no collision free
balanced $4\times4$ matrix.  
There cannot be a collision free balanced binary $G\times 4$ matrix
with $G\ge6$ either.  There are only ${4\choose 2}=6$
distinct such rows and using them all would bring collisions.
Similarly, there are no $B\times 4$ collision free balanced binary
matrices.

When $B=G=6$, it is possible to avoid collisions. For instance, we
could use the matrix
\begin{align}\label{eq:bg6}
\begin{pmatrix}
+     & +    & +     & \cdot & \cdot & \cdot \\
+     & +    & \cdot & +     & \cdot & \cdot \\
+     & \cdot& \cdot & \cdot & +     & +     \\
\cdot & +    & \cdot & \cdot & +     & +     \\
\cdot & \cdot& +     & +     & +     & \cdot \\
\cdot & \cdot& +     & +     & \cdot & +     \\
\end{pmatrix},
\end{align}
which has no collisions.     For $B=G=8$ we could use
\begin{align}\label{eq:bg8}
\begin{pmatrix}
  + &  + &  + &  + &  \cdot &  \cdot &  \cdot &  \cdot \\
  + &  + &  \cdot &  \cdot &  \cdot &  \cdot &  + &  + \\
  + &  \cdot &  + &  \cdot &  + &  + &  \cdot &  \cdot \\
  + &  \cdot &  \cdot &  + &  \cdot &  + &  + &  \cdot \\
  \cdot &  + &  + &  + &  + &  \cdot &  \cdot &  \cdot \\
  \cdot &  + &  \cdot &  \cdot &  + &  + &  \cdot &  + \\
  \cdot &  \cdot &  + &  \cdot &  + &  \cdot &  + &  + \\
  \cdot &  \cdot &  \cdot &  + &  \cdot &  + &  + &  + \\
\end{pmatrix}.
\end{align}
\end{proof}

The first three columns of the matrix in~\eqref{eq:bg8} are the same as
in a classical $2^3$ factorial design. That matrix is not such
a design, and indeed that design would not have balanced rows.

Next we consider how to create larger $G\times B$ balanced binary
collision free matrices from smaller ones.
\begin{theorem}
Let $X\in\{-1,1\}^{B\times G}$ be a balanced binary matrix
  with no collisions.
Then there is a balanced binary matrix $\tilde X\in\{-1,1\}^{(B+4)\times (G+4)}$
with no collisions.
\end{theorem}
\begin{proof}
Let $r_1$ and $r_2$ be the first two rows of $X$, 
let $c_1$ and $c_2$ be the first two columns of $X$
and choose $z\in\{-1,1\}$.
Now let
\begin{align}\label{eq:growit}
X^* = \begin{pmatrix}
\phm X &  \phm c_1 & -c_1 & \phm c_2 & -c_2\\
\phm r_1 & \phm z& \phm z & -z & -z\\
-r_1 & \phm z& \phm z & -z & -z\\
\phm r_2 & -z & -z & \phm z & \phm z\\
-r_2 & -z & -z & \phm z & \phm z\\
\end{pmatrix}.
\end{align}
Every row and every column of $X^*$ is balanced
by construction. There are no collisions
among the first $B$ rows or first $G$ columns of $X^*$ because
there are none in $X$. 

Now we consider the last four rows of $X^*$.
Row $B+1$ does not collide with the last two rows
because $r_1$ does not collide with $r_2$. Rows $B+1$
and $B+2$ are opposite in their first $G$ columns
but they agree in the next $4$ columns so they do not collide.
By symmetry, this argument shows that there are no collisions
among the last four rows of $X^*$ or among the last four
columns.

It remains to check whether any of the new rows (or columns)
collide with any of the old ones.
Row $B+1$ of $X^*$ cannot collide with row $k$ of $X^*$ for any
$1<k\le B$ because $r_1$ does not collide with any of the corresponding
rows of $X$.
Rows $B+1$ and $1$ of $X^*$ agree in the first $G$ columns
but differ in exactly two of the last $4$ columns of $X^*$
so they do not collide.
Therefore row $B+1$ of $X^*$ does not collide with any of
the first $B$ rows. Row $B+2$ of $X^*$ equals $\tilde r_1$ in its
first $G$ columns.
Therefore it cannot collide with row $k$ of $X^*$ for any
$1<k\le B$. By construction it matches row $1$ in two of the new
columns and is opposite row $1$ in the other two.  It follows that
none of the last four rows of $X^*$ collide with any of the
first $B$ rows. By symmetry there are no collisions among any
of the last four columns of $X^*$ and any of the first $G$
columns.
\end{proof}

The Theorem above gives an approach to creating design
matrices. We start with a small
matrix and grow it by repeatedly applying equation~\eqref{eq:growit}. 
It is not necessary to grow $X$ via the first two
rows and columns.  It would work to choose any two distinct
rows or columns. For instance they could be chosen randomly
or chosen greedily to optimize some property of the resulting matrix.

Repeatedly applying equation~\eqref{eq:growit}
will give a nearly square matrix because it keeps
adding $4$ to both the number of rows and the number
of columns. We might want to have $G\gg B$.

We can grow the matrix by $4$ rows and $8$ columns via
\begin{align}\label{eq:growitmore}
X^* = \begin{pmatrix}
 X & \phm c_1 & -c_1 & \phm c_2 & -c_2 & \phm c_3 & -c_3& \phm c_4 & -c_4\\
\phm r_1 & \phm z_1& \phm z_1 & -z_1 & -z_1& \phm z_2& \phm z_2 & -z_2 & -z_2\\
-r_1 & \phm z_1& \phm z_1 & -z_1 & -z_1 &\phm z_2& \phm z_2 & -z_2 & -z_2\\
\phm r_2 & -z_1 & -z_1 & \phm z_1 & \phm z_1& -z_2& -z_2 & \phm z_2 & \phm z_2\\
-r_2 & -z_1 & -z_1 & \phm z_1 & \phm z_1& -z_2& -z_2 & \phm z_2 & \phm z_2\\
\end{pmatrix},
\end{align}
for any $z_1,z_2\in\{-1,1\}$ where $r_1$ and $r_2$ are
any two rows of $X$ and $c_1,\dots,c_4$ are any four
columns of $X$.
Equation~\eqref{eq:growitmore} adds four rows and eight
columns. The same idea could extend a $G\times B$ matrix
to a $3G\times (B+4)$ matrix,
tripling the number of columns (GEOs) while adding 
only four rows (brands).

The methods of this section show that there are some large collision
free designs.  
We find that starting with the matrix~\eqref{eq:bg6}
or~\eqref{eq:bg8} and growing it by repeatedly applying~\eqref{eq:growit}
yields designs that include some correlations very close to $\pm1$.
The scrambled checkerboard approach
tends to produce designs with smaller maximum absolute
correlation than the growth approach.
Also, numerically searching with that algorithm turns up $8\times 8$
designs but not $6\times 8$ designs, which we suspect do not exist.

\section{Regression results}\label{sec:regression}

The regression model~\eqref{eq:singlemodel} was simulated
with advertising effectiveness $\beta=5.0$ in $20$ GEOs
of which $10$ had increased spend equal to $1\%$ of the prior period's
sales.  
Further details are in Section~\ref{sec:simulation}.
Figure~\ref{fig:onesinglesim} shows one realization.
The simulation was done $1000$ times in total.
Then, using the same random seeds, the simulation was repeated with
increased spend of $0.5$\% instead of $1$\%.

\begin{figure}[t]
\centering
\includegraphics[width=0.99\hsize]{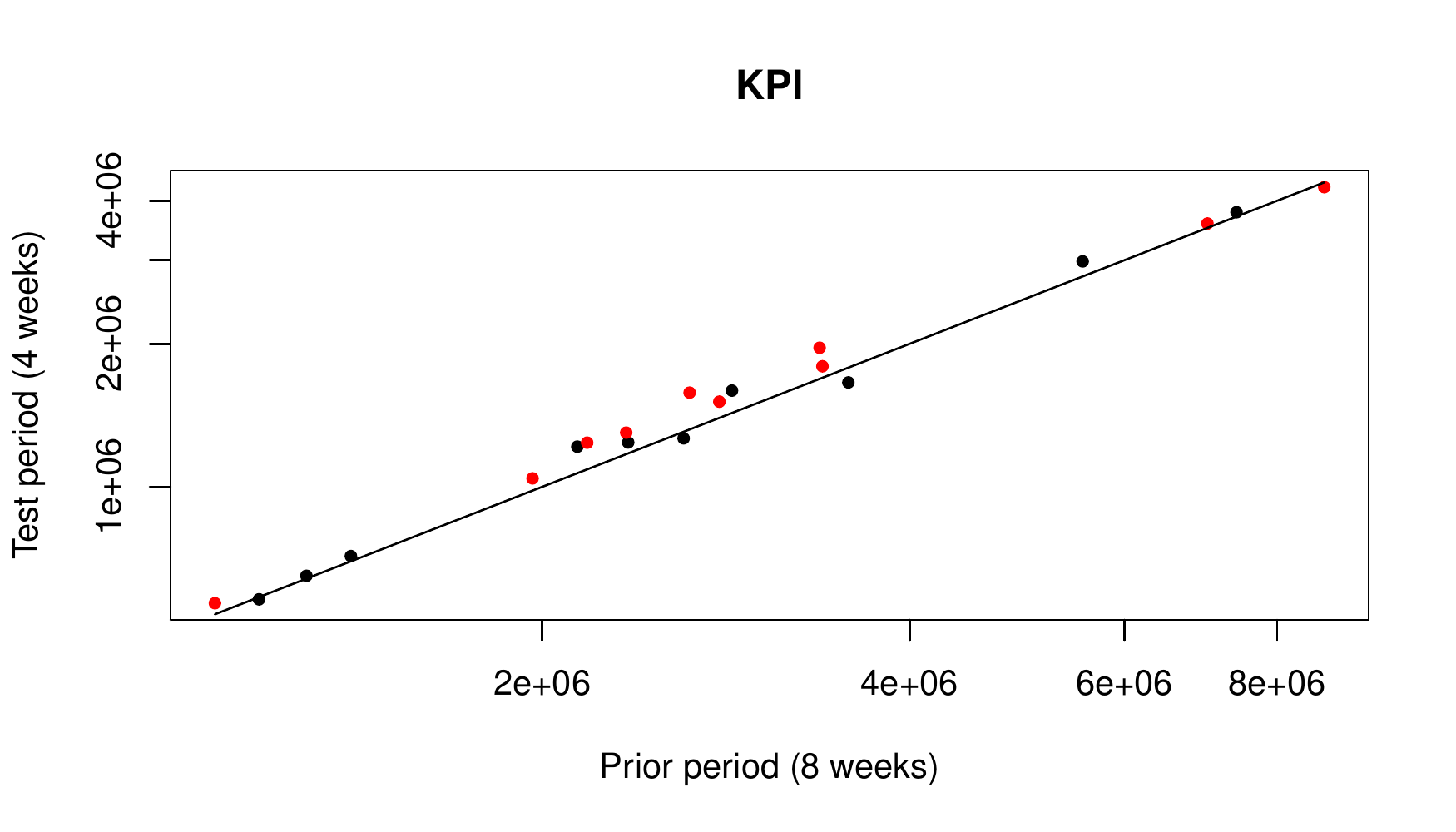}
\caption{\label{fig:onesinglesim}
One realization of a single brand simulation. 
Treatment GEOs are in red, control in black.
The reference line is at $y=x/2$ because the
test period has half the length of the prior period.
}
\end{figure}

For each simulated data set, weighted least squares regression
was used. The weights were proportional to $(1/Y^\pre)^2$, making
them inversely proportional to variance. Unweighted regression 
does not give reliable confidence intervals and $p$-values in
this setting.

Some results are plotted in Figure~\ref{fig:single}
and some numerical summaries are in Table~\ref{tab:singleoutput}.
The top panels of Figure~\ref{fig:single} 
show histograms of $1000$ two-sided $p$-values
for $H_0{:}\beta=0$. This null was rejected $28.6$\% of the
time for the experiment with a smaller treatment size,
and $81.3$\% of the time for 
the one with a larger treatment size.
The middle panels show histograms of twice the standard
error of $\hat\beta$, roughly the distance from $\hat\beta$
to the edge of a $95$\% confidence interval. 
At $0.5$\% treatment this uncertainty
averaged $6.68$ while at $1$\% it averaged $3.34$, just over $66$\%
of the true value $5$.
The root mean squared error in $\hat\beta$ was $3.21$ for
small treatment differences and $1.61$ for large ones.
The bottom panels show histograms of the estimates $\hat\beta$
for only those simulations in which $H_0$ was rejected at the $5$\% level.
The average estimated effect was $8.64$ for the smaller treatment size
and $5.51$ for the larger one.

\begin{figure}[t]
\centering
\includegraphics[width=0.99\hsize]{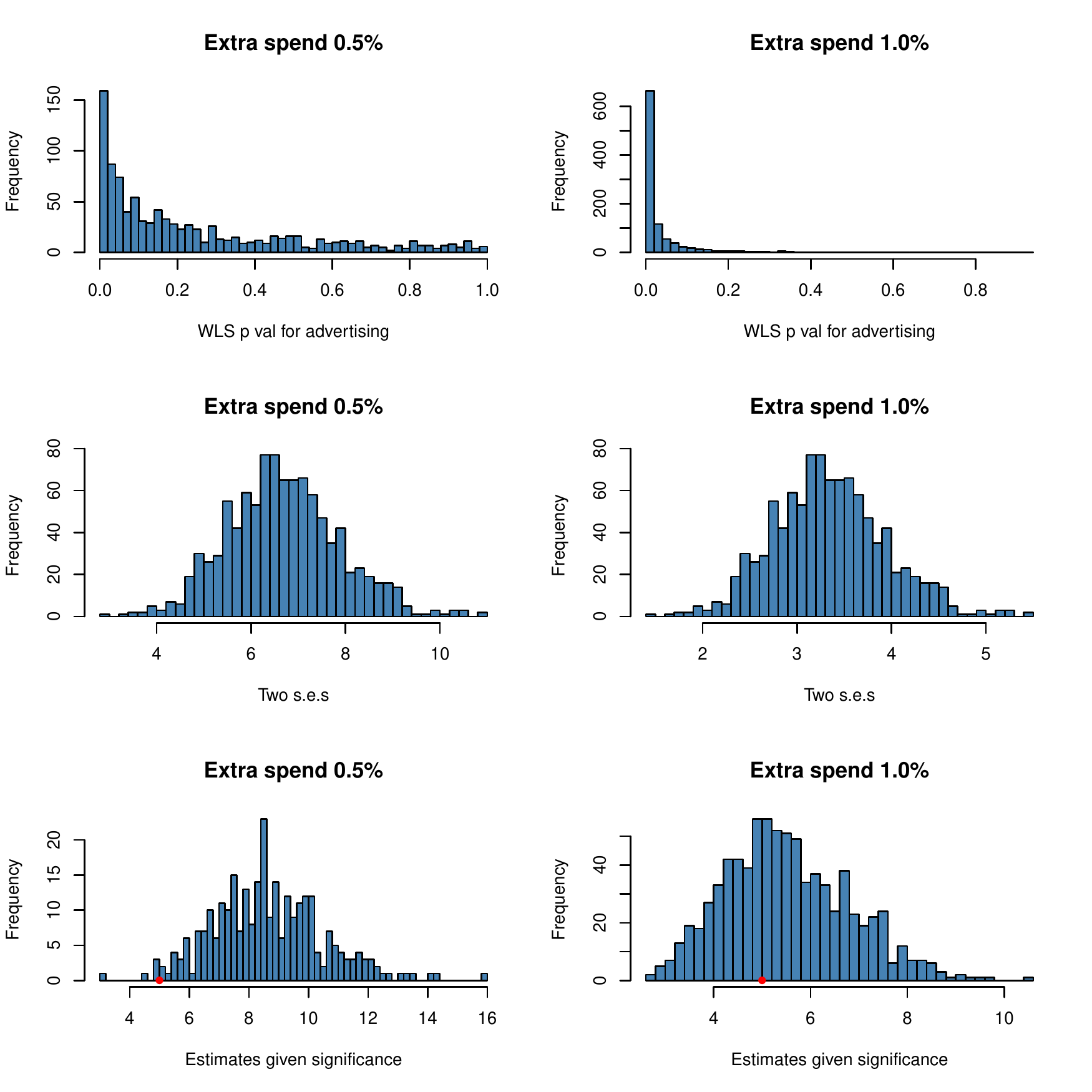}
\caption{\label{fig:single}
Results from repeated single brand simulations.
}
\end{figure}

The smaller treatment size has very low power, very wide confidence
intervals, and in those
instances where it detects an advertising effect, it gives a
substantial overestimate of effectiveness.
The larger treatment size has greater power and only slight overestimation
of $\beta$ when it is significant. But it still
yields a wide confidence interval for $\beta$.

\begin{table}
\centering
\begin{tabular}{lcccc}
\toprule
Trt & $\wh\Pr(p\le 0.05)$ & $2\mathrm{se}(\hat\beta)$ &$\hat\e((\hat\beta-\beta)^2)^{1/2}$ & $\hat\e(\hat\beta\mid p\le 0.05)$\\
\midrule
0.5\%  &0.29 & 6.68 &3.21&      8.64\\
1.0\%  &0.81 & 3.34 &1.61&      5.51\\
\bottomrule
\end{tabular}
\caption{
\label{tab:singleoutput}
Output summary of $1000$ simulations of the single
brand experiment.
}
\end{table}

\section{Multibrand experimental results}\label{sec:multibrandsimu}

The multibrand setting was simulated with $B=30$ brands
over $G=20$ GEOs.  Treatment versus control was
assigned with scrambled checker designs from 
Section~\ref{sec:design}.
The effectiveness of brand $b$ was generated
from $\beta_b\sim\dnorm(5,1)$, so advertising returns
are usually in the range from $3$ to $7$.

\subsection{Shrinkage estimation of $\beta_b$}\label{sec:shrinkage}
We write $\hat\beta_b$ for the
least squares estimate of $\beta_b$ from
brand $b$ data and set
$\hat{\bar\beta} = (1/B)\sum_{b=1}^B\hat\beta_b$.
We can estimate $\beta_b$ by a shrinkage estimator
formed as a weighted average of $\hat\beta_b$
and $\hat{\bar\beta}$.
\citet[Section 4]{xie:kou:brow:2012} propose estimators
of the form
\begin{align}\label{eq:xieshrink}
\tilde\beta_b = \frac{\lambda}{\var(\hat\beta_b)+\lambda}\hat\beta_b
+\frac{\var(\hat\beta_b)}{\var(\hat\beta_b)
+\lambda}\hat{\bar\beta}
\end{align}
for a parameter $\lambda$ that must be chosen. The larger $\lambda$
is, the more emphasis we put on brand $b$'s own data instead of
the pooled data. For brands with large $\var(\hat\beta_b)$, more
weight is put on the pooled estimate $\hat{\bar\beta}$.
Xie et al.'s~(2012) \nocite{xie:kou:brow:2012}
main innovation
is in shrinkage methods for data of unequal variances as we have here.
To use their method we replace $\var(\hat\beta_b)$
by unbiased estimates $\wh\var(\hat\beta_b)$ taken from the linear model output,
and choose $\lambda$.

\cite{xie:kou:brow:2012} give theoretical support for choosing
$\lambda$ to minimize the following unbiased estimate of
the expected sum of squared errors
\begin{align}
\begin{split}\label{eq:sureg}
\mathrm{SURE}^G(\lambda)
&= 
\frac1B\sum_{b=1}^B\frac{ \var(\hat\beta_b)^2}{(\var(\hat\beta_b)+\lambda)^2}
(\hat\beta_b-\hat{\bar\beta})^2\\
&\qquad
+\frac1B\sum_{b=1}^B\frac{ \var(\hat\beta_b)}{\var(\hat\beta_b)+\lambda}
\Bigl(\lambda - \var(\hat\beta_b)+\frac2B\var(\hat\beta_b)\Bigr).
\end{split}
\end{align}
This function is not convex in $\lambda$ but a practical way
to choose $\lambda$ is to evaluate 
$\mathrm{SURE}^G$ on a grid of, say $1001$, $\lambda$ values.
Letting the typical weight on $\hat\beta_b$ take values
$u\in\{0,1/1000,2/1000,\dots,1\}$
we use 
$$\lambda = \frac1B\sum_{b=1}^B\var(\hat\beta_b)\times \frac{u}{1-u}$$
where $u=1$ means $\lambda=\infty$ which simply means
$\tilde\beta_b=\hat\beta_b$.

We can measure the efficiency gain from shrinkage via
$$
\mathrm{Eff}=\frac
{ \frac1B\sum_{b=1}^B (\hat\beta_b-\beta_b)^2}
{ \frac1B\sum_{b=1}^B (\tilde\beta_b-\beta_b)^2}.
$$
Figure~\ref{fig:efficiency} shows a histogram of this
efficiency measure in $1000$ simulations. On average
it was about $3.17$ times as efficient to use shrinkage
when the experimental treatment is to increase advertising
by $1$\% of prior sales. For smaller experiments, at $0.5$\%
of sales, the average efficiency gain was $7.82$.
For each given brand $b$, the information from $B-1$ 
other brands' data yields a big improvement in accuracy.
Recall that $\beta_b\simiid \dnorm(5,1)$.
The gain from shrinkage would be less if the underlying $\beta_b$ were
less similar and greater if they were more similar.

\begin{figure}%
\centering
\subfloat[Experimental spend $1$\%.\label{subfig:1pc}]{%
\includegraphics[width=0.49\textwidth]{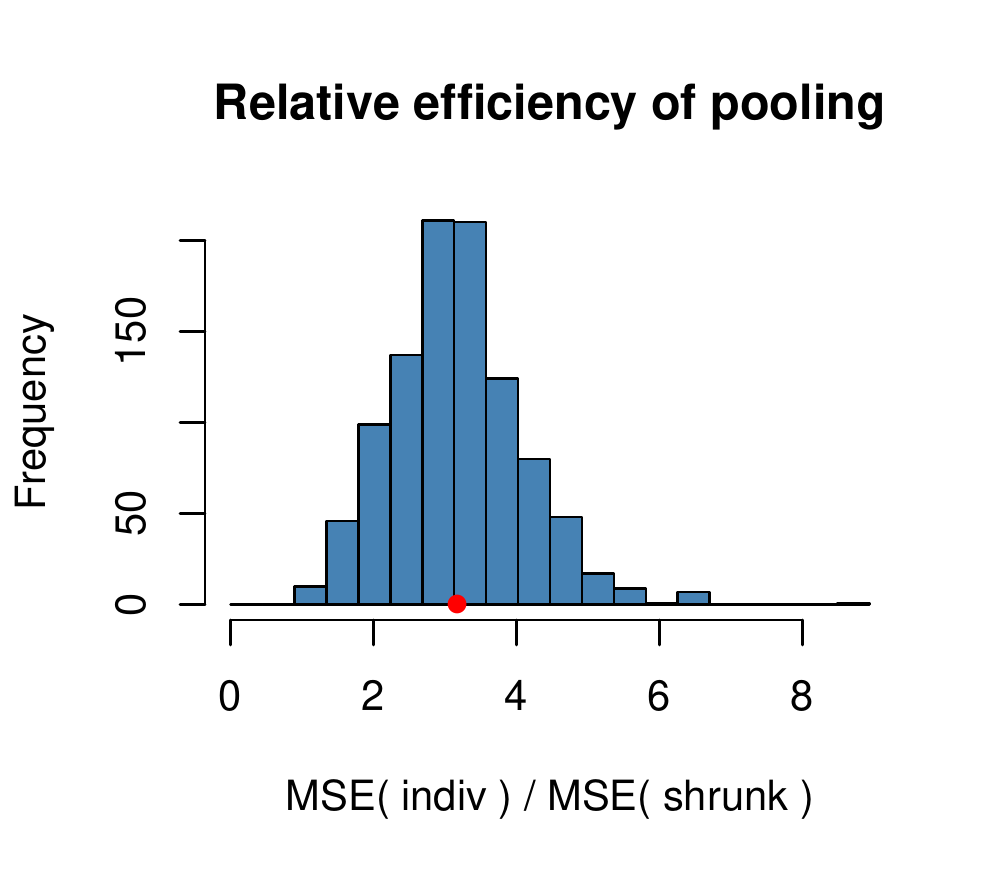}
}
\hfill
\subfloat[Experimental spend $0.5$\%.\label{subfig:0p5pc}]{%
\includegraphics[width=0.49\textwidth]{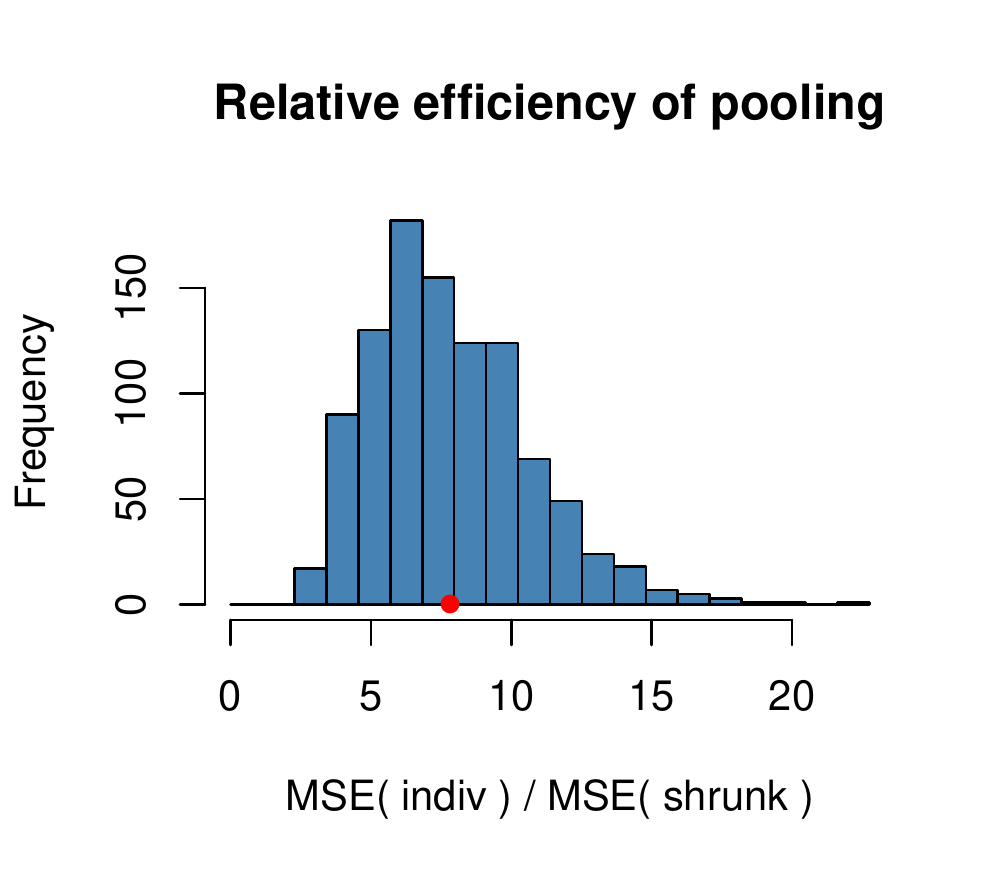}
}
\caption{Relative efficiency of shrinkage
estimates compared to single brand regressions.
}
\label{fig:efficiency}
\end{figure}

\subsection{Average return to advertising}\label{sec:avgreturn}
The quantity $\bar\beta$ measures the overall 
return to advertising averaged over all brands. Although
individual returns $\beta_b$ are more informative, their average
can be estimated much more reliably. In small experiments where
some individual $\hat\beta_b$'s are not well determined
it may be wiser to base decisions on $\hat{\bar{\beta}}$.

\begin{figure}[t]%
\centering
\subfloat[Experimental spend $1$\%.\label{subfig:1pc}]{%
\includegraphics[width=0.49\textwidth]{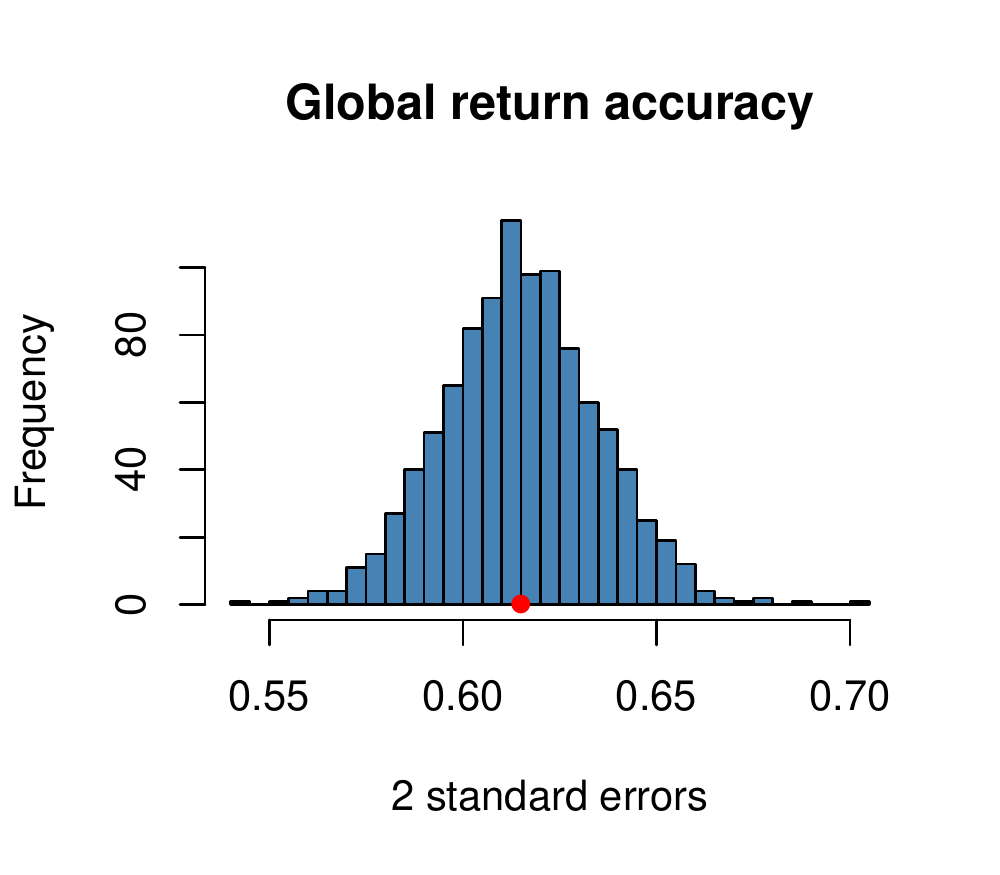}
}
\hfill
\subfloat[Experimental spend $0.5$\%.\label{subfig:0p5pc}]{%
\includegraphics[width=0.49\textwidth]{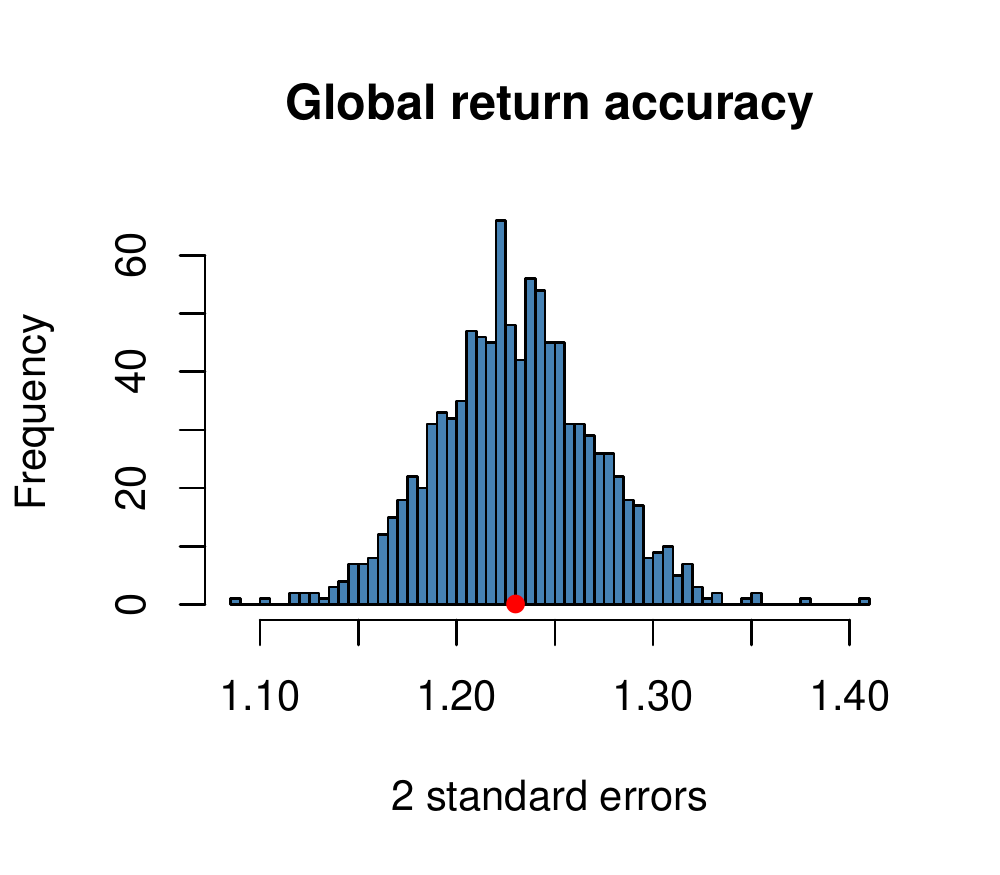}
}
\caption{Two standard errors of $\hat{\bar\beta}$.
}
\label{fig:twose}
\end{figure}

We can estimate $\bar\beta$ by 
$\hat{\bar\beta} = (1/B)\sum_{b=1}^B\hat\beta_b$
and then using the individual regressions compute
$\wh\var(\hat{\bar\beta})=B^{-2}\sum_{b=1}^B\wh\var(\hat\beta_b)$.
Figure~\ref{fig:twose} shows histograms of
$2(\wh\var(\hat{\bar\beta}))^{1/2}$.
Table~\ref{tab:twose} compares average
values of twice the standard
error for $\hat\beta$ in a single brand experiment
with twice the standard error for $\hat{\bar\beta}$
in a multibrand experiment. As we might expect
the multibrand standard errors are roughly $\sqrt{B}=\sqrt{30}$
times smaller. Similarly, doubling the spend roughly
halves the standard error.  

\begin{table}
\centering
\begin{tabular}{cccc}
\toprule
\multicolumn{2}{c}{1\% spend}&\multicolumn{2}{c}{0.5\% spend}\\
\midrule
Single & Multiple & Single & Multiple \\
\midrule
$3.34$ & $0.62$ & $6.68$ & $1.23$ \\
\bottomrule
\end{tabular}
\caption{Average over simulations of two standard
errors for $\hat\beta$ (single brands) and
$\hat{\bar\beta}$ (multiple brands).
}
\label{tab:twose}
\end{table}


\section{Simulation details}\label{sec:simulation}

In each simulation, the design was generated by 
the scrambled checker algorithm described in Section~\ref{sec:design}.
Then the data were sampled from the Gamma distributions described
here.

\subsection{Gamma distributions}
The Gamma distribution 
has a standard deviation 
proportional to its mean, matching a pattern
in the real sales data. 
When the shape parameter is $\kappa>0$ the Gamma
probability density function is
${x^{\kappa-1}e^{-x}}/{\Gamma(\kappa)}$ for $x>0$.
To specify a scale parameter, we multiply $X\sim\dgam(\kappa)$ by
the desired scale $\theta$.

The random variable $\theta X$ has mean $\kappa\theta$
and variance $\kappa\theta^2$, leading to a coefficient
of variation equal to $1/\sqrt{\kappa}$ for any $\theta$.
The shape $\kappa=1/\cv^2$ yields a Gamma random variable
with the desired coefficient of variation.

The coefficient of variation for the average
of $n$  observations from one GEO (e.g., $n_\post=4$ in the test
period and $n_\pre=8$ in the background period)
is approximately $1/\sqrt{n}$ times the coefficient of variation
of a single observation. The coefficient of variation
for single observations from a set of $8$ week trial
periods was about $0.15$ while that for $4$ week followup
periods was about $0.1$.  These figures are based
on aggregates over GEOs that were very similar for 
all of the different brands.
An $8$ week trial period has within it more
seasonality than a $4$ week period has, and
so it is reasonable that we would then 
measure a larger coefficient of variation.

To simulate with a specific coefficient of variation
we use $\kappa = n/\cv^2$. This leads to
$\kappa_\pre = 8/0.15^2\doteq 356$ and
$\kappa_\post = 4/0.1^2 =400$.

Gamma random variables are never negative
which gives them a further advantage over simulations
with Gaussian random variables.
For the specific parameter choices above, the shape
parameters are large enough that the Gamma random variables
are not strongly skewed (their skewness is $2/\sqrt{\kappa}$). 
A Gaussian distribution might
give similar results. The Gamma distribution is useful because
it can be used to simulate either strongly or mildly skewed 
data that are always nonnegative.

\subsection{Data generation}
For a single brand experiment, the data
are generated as follows.
First the underlying sizes of the GEOs
were sampled as $S_g = 10^{7-U_g}$
where $U_g\sim\dustd(0,\phi)$, for $g=1,\dots,G$.
The quantity $S_g$ is interpreted as a size measure
for GEO $g$.  We will use it as the expected
prior sales, which is then roughly proportional
to the number of customers in GEO $g$.
Choosing $\phi =1$ means that we consider GEOs ranging in size
by a factor of about $10$ from largest to smallest.

The prior KPIs are generated as
$$
Y_g^\pre \simind S_g \times {\dgam(\kappa_\pre)}/{\kappa_\pre},
$$
where $\kappa_\pre = n_\pre/\cv_\pre^2$.
Then $\e(Y_g^\pre)=S_g$. 
Let the spending level in the experimental period
be $X_g^\post$ in GEO $g$. 
Then the KPI in the experimental period is
generated as
$$
Y_g^\post \simind \frac{n_\post}{n_\pre}\times S_g\times\dgam(\kappa_\post)/\kappa_\post + X_g^\post\beta,
$$
where $\kappa_\post = n_\post/\cv_\post^2$.
The factor $n_\post/n_\pre$ adjusts for different sizes of
prior and experimental observation windows.
The term $X_g^\post\beta$ is the additional KPI attributable
to advertising.

\subsection{Checkerboard designs}\label{sec:checkerdone}

The scrambled checkerboard design was run
for $2\times G\times B\times 25=30{,}000$ steps. The expected
number of flips for each pixel in the image is $25$.
This is much more than the number at which the root mean
squared correlations stabilize.

\section{A fully Bayesian approach}\label{sec:fullbayes}
Stein shrinkage is an empirical Bayes approach. Here we consider a fully Bayesian alternative.
When it comes to pooling information together from observations that arise from a common model but corresponding to different sets of parameters, Bayesian hierarchical models arise as a natural solution.

One advantage of the Bayesian approach is that it allows us to present the uncertainty in our estimates.  For each brand $b$, we can get an interval $(L_b,U_b)$ such that
$\Pr( L_b \le \beta_b \le U_b\mid \text{data} ) = 0.95$ 
without making any (additional) assumptions.
These posterior credible intervals are easier to compute than confidence
intervals from Stein shrinkage.  We can also use posterior credible intervals
at the planning stage. To do that, we simulate the data several times and
record how wide the posterior credible intervals are.  If they are too wide
we might add more GEOs or increase the differential spend $\delta$.

\subsection{A hierarchical model}
We consider model \eqref{eq:multimodel} in a Bayesian context, which translates as follows:
\begin{align}
  \mu_{gb} &:= \alpha_{0b} + \alpha_{1b}Y^\pre_{gb}+ \beta_b X_{gb}^\post, \quad b=1,\dots,B,\ g=1,\dots,G  \label{eq:hierarchicalprior1} \\
  Y^\post_{gb} &\sim \dnorm\Bigl(\mu_{gb}, \bigl({\sigma_b}/{Y^{\pre}_{gb}}\bigr)^2\Bigr), \quad b=1,\dots,B,\ g=1,\dots,G  \label{eq:hierarchicalprior2} \\
  \sigma_b^2 &\sim \mathcal{IG}(10^{-3}, 10^{-3}), \quad b=1,\dots,B  \label{eq:hierarchicalprior3} \\
  \beta_b &\sim \dnorm(\beta, \sigma_{\beta}^2), \quad b=1,\dots,B  \label{eq:hierarchicalprior4} \\
  \sigma_{\beta}^2 &\sim \mathcal{IG}(0.5, 0.5),\quad\text{and},  \label{eq:hierarchicalprior5} \\
  \beta &\sim {\boldsymbol{1}}_{\mathbb{R}}.  \label{eq:hierarchicalprior6}
\end{align}

Definitions \eqref{eq:hierarchicalprior1}, \eqref{eq:hierarchicalprior2} and \eqref{eq:hierarchicalprior3} the mirror model \eqref{eq:multimodel} that we described earlier. Notice that we signal the weighted regression explicitly in \eqref{eq:hierarchicalprior2}. The hierarchical Gaussian prior \eqref{eq:hierarchicalprior4} on the coefficients \(\beta_b\) involves two new hyperparameters \(\beta\) and \(\sigma_{\beta}^2\) which respectively represent the overall mean and variance of all returns \(\beta_b\). Since we assumed in the beginning that all brands had somewhat similar returns, we choose a semi-informative prior \eqref{eq:hierarchicalprior5} on \(\sigma_{\beta}^2\) that favors plausible, not too large, values. For our simulations we use a flat prior \eqref{eq:hierarchicalprior6} on \(\beta\) relying on the data to drive the inference. One could also use a Gaussian prior for \(\beta\), crafting its mean and variance based on the prior knowledge of the brands at hand.

\subsection{Simulation details}\label{sec:bayesdetails}
The data were generated according to the procedures described in Section \ref{sec:simulation}. Samples were collected from the posterior distribution using STAN software \citep{stan:2016}.

We simulated many different conditions and consistently found that
the Stein and Bayes estimates were close to each other. In this section we
present just one simulation matching parameters of interest to some of
our colleagues at Google. We consider $G=160$ GEOs,  only $B=4$
brands and we take advertising effectiveness $\beta_b$ to be $\dnorm(1,1)$.  
Using $\e(\beta_b)=1$ produces a setting where a dollar of advertising typically brings
back a dollar of sales in the observation period. That implies a short
term loss with an expected longer term benefit from adding or retaining customers.
Taking the standard deviation of $\beta_b$ to one implies very large
brand to brand variation. We still see a benefit from pooling only $4$
brands as diverse as that.
The amount of extra spend is set to $1$\% of prior sales ($\delta=0.01$).
We repeated this simulation $1000$ times.

In this setting with $160$ GEOs and $4$ brands there will
always be some GEOs that get the exact same treatment for all $4$
brands.

\subsection{Agreement with shrinkage estimates}

Figure~\ref{fig:bayes_error_diff} 
compares the RMSE $[(1/B)\sum_{b=1}^B(\hat\beta_b-\beta_b)^2]^{1/2}$
for Stein and Bayes estimation in $1000$
simulations with $B=4$. The methods have very similar accuracy.
For high brand to brand standard deviation $\sigma_b=1.0$, there is a
slight advantage to Bayes.  For lower brand to brand standard deviation $\sigma_b=0.25$,
there is a small advantage to Stein.  The Bayesian estimate was at a disadvantage
there because the prior variance was $1/\mathrm{Gamma}(0.5,0.5)$ giving $\sigma_b$
a median of about $1.48$. This shows that the Bayesian estimate is not overly
sensitive to our widely dispersed prior distribution on $\sigma_b$.

\begin{figure}[t!]%
\centering
\subfloat[Stein versus Bayes RMSE, $\sigma_b=1$.\label{subfig:bayes_error}]{%
\includegraphics[width=0.46\textwidth]{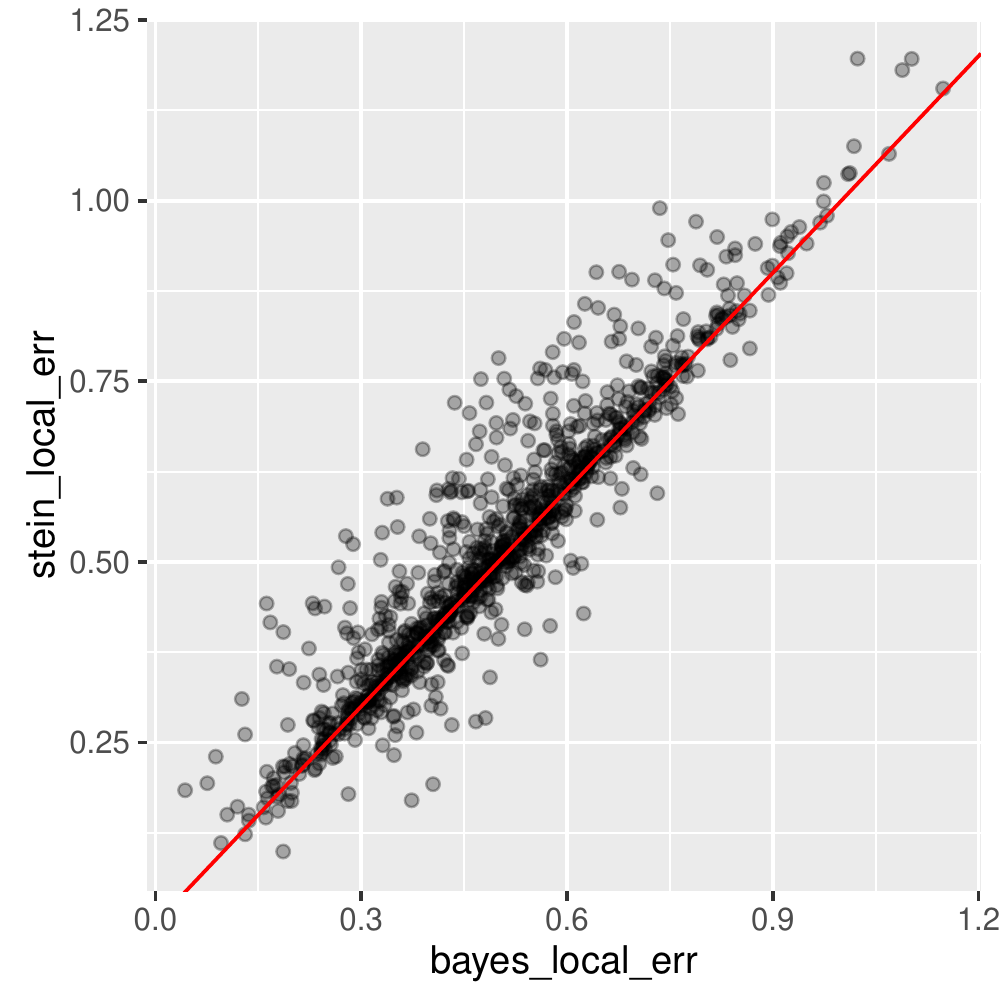}}
\subfloat[Stein minus Bayes RMSE, $\sigma_b=1$.\label{subfig:bayes_error_diff}]{%
\includegraphics[width=0.46\textwidth]{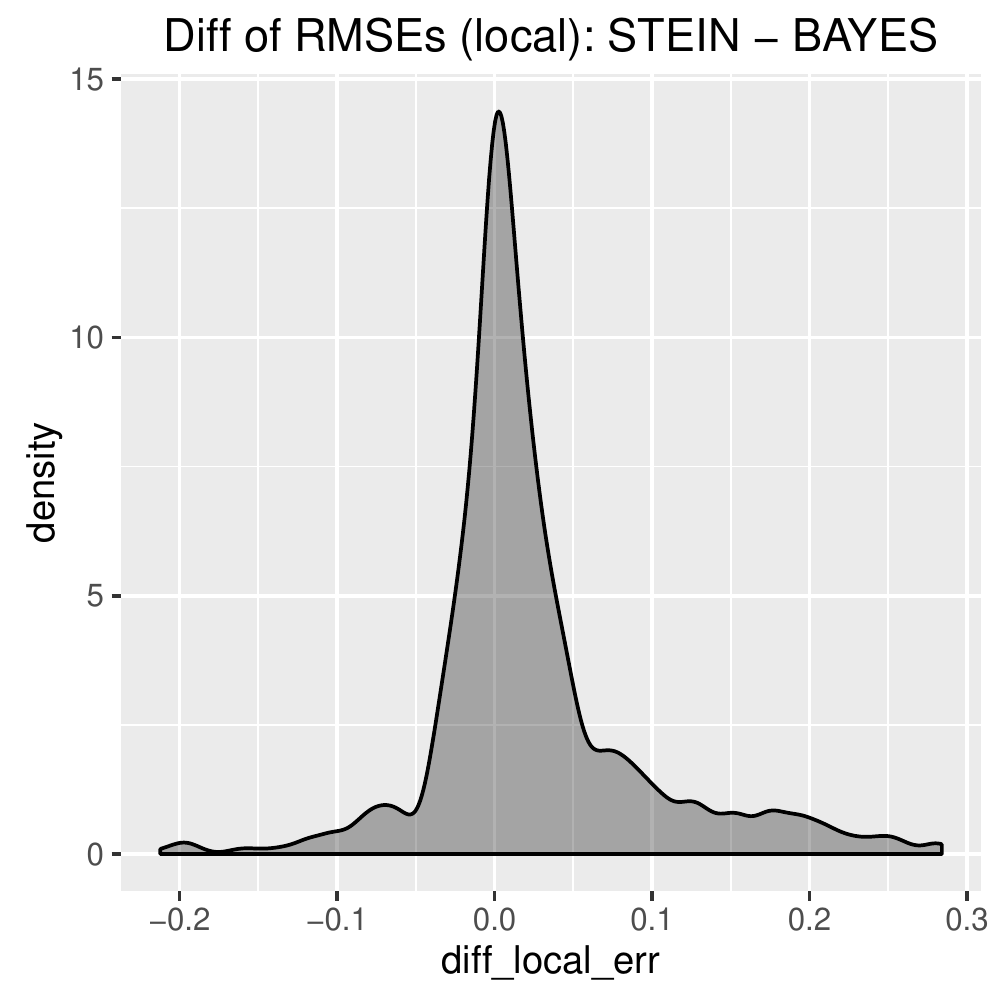}}
\hfill
\subfloat[Stein versus Bayes RMSE, $\sigma_b=0.25$.\label{subfig:bayes_error}]{%
\includegraphics[width=0.46\textwidth]{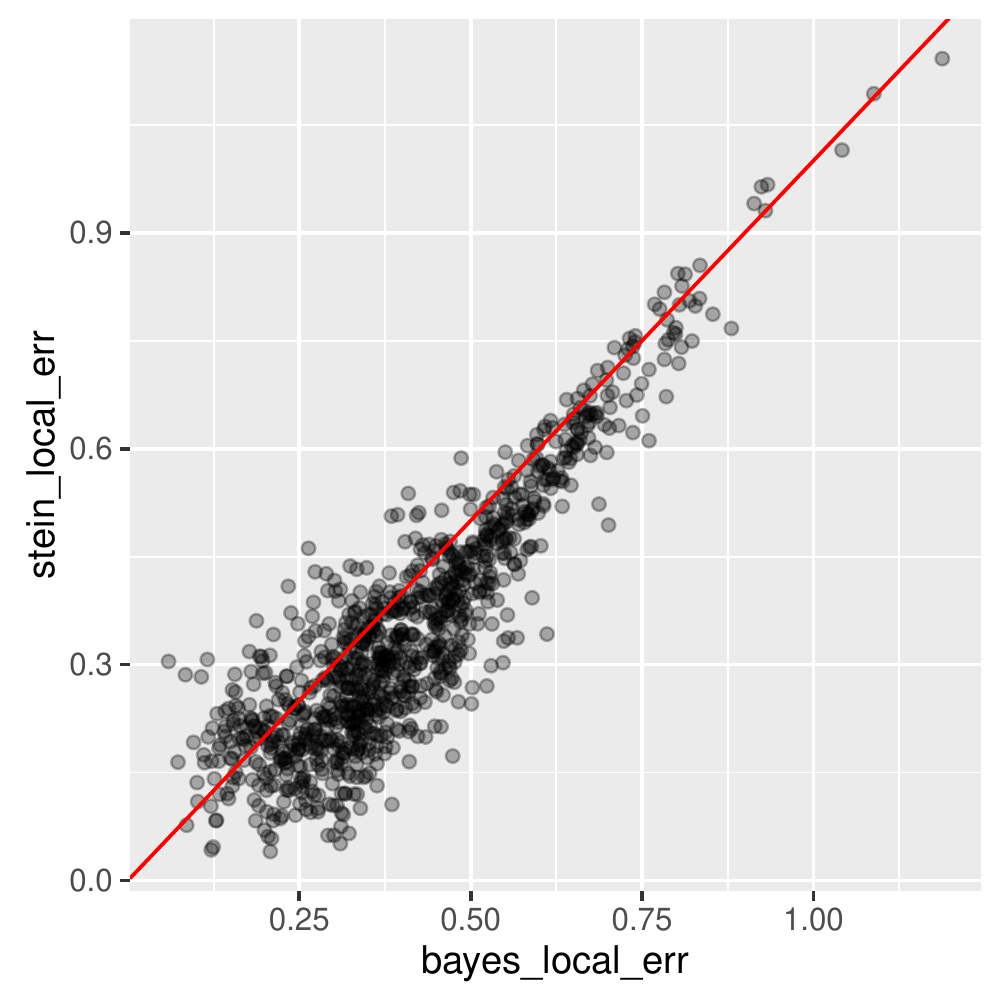}}
\subfloat[Stein minus Bayes RMSE, $\sigma_b=0.25$.\label{subfig:bayes_error_diff}]{%
\includegraphics[width=0.46\textwidth]{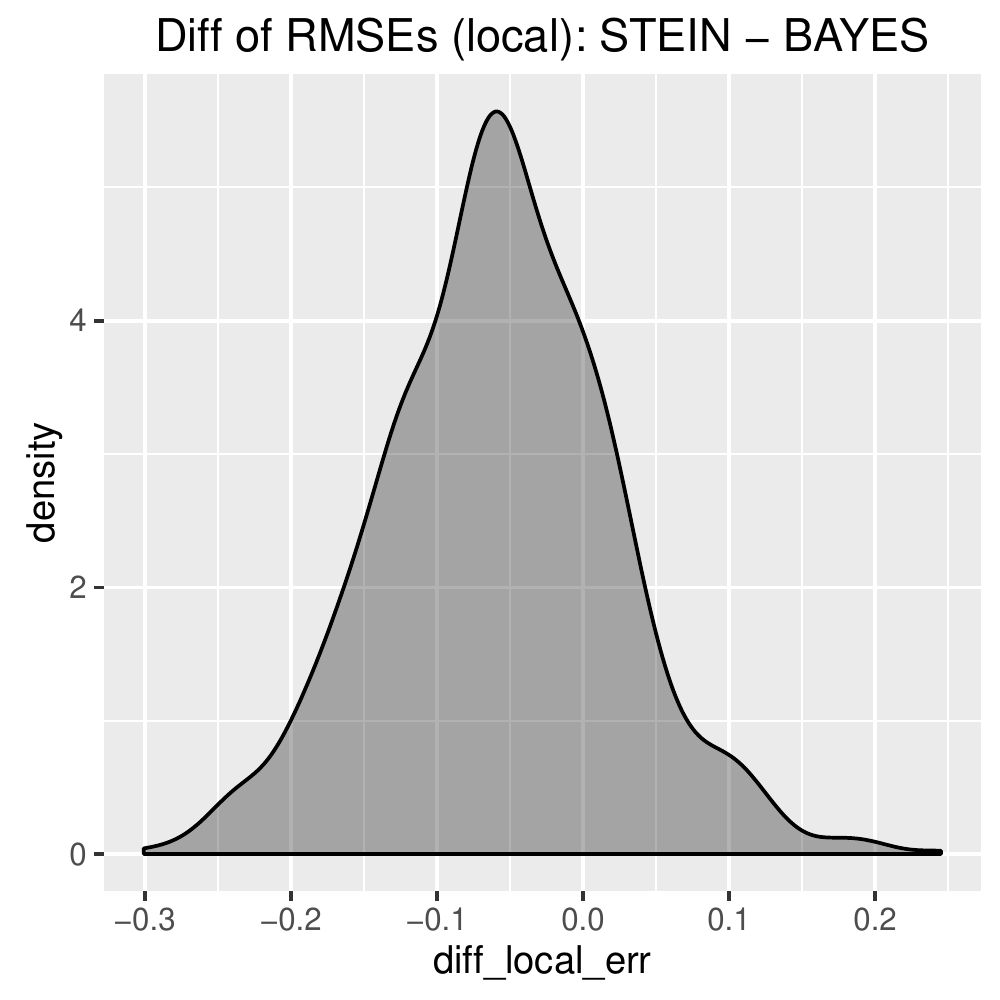}}
\caption{Comparison of Bayes and Stein RMSEs on two simulations
  of $1000$ replicates with $B=4$ brands.}
\label{fig:bayes_error_diff}
\end{figure}

\subsection{Posterior credible intervals}


To investigate what power can be gained pooling data together using a multibrand experiment over conducting multiple single-brand experiments independently, we simulated datasets following the same procedure as before, using \(G = 160\) GEOs with either \(B = 1\) (no pooling) or \(B = 4\) (pooling) brands, the effectivenesses of which were drawn from a Gaussian distribution \(\dnorm(1, \sigma^2_b)\) with $\sigma_b=0.25$. Now the brands return on average one dollar of incremental revenue per dollar spent, and the standard deviation of $0.25$ represents substantial brand differences. The relative incremental ad-spend we made varies from 0.5\% to 2\% to show how it impacted the results. Figure \ref{fig:bayes_pool_vs_nopool} displays the estimated densities (over 10{,}000 replications) of the half-width of 95\% credible intervals around the brands' effectivenesses, the solid lines representing the 95\% quantile of these densities, in each scenario.

\begin{figure}[t]%
\centering
\includegraphics[width=0.98\textwidth]{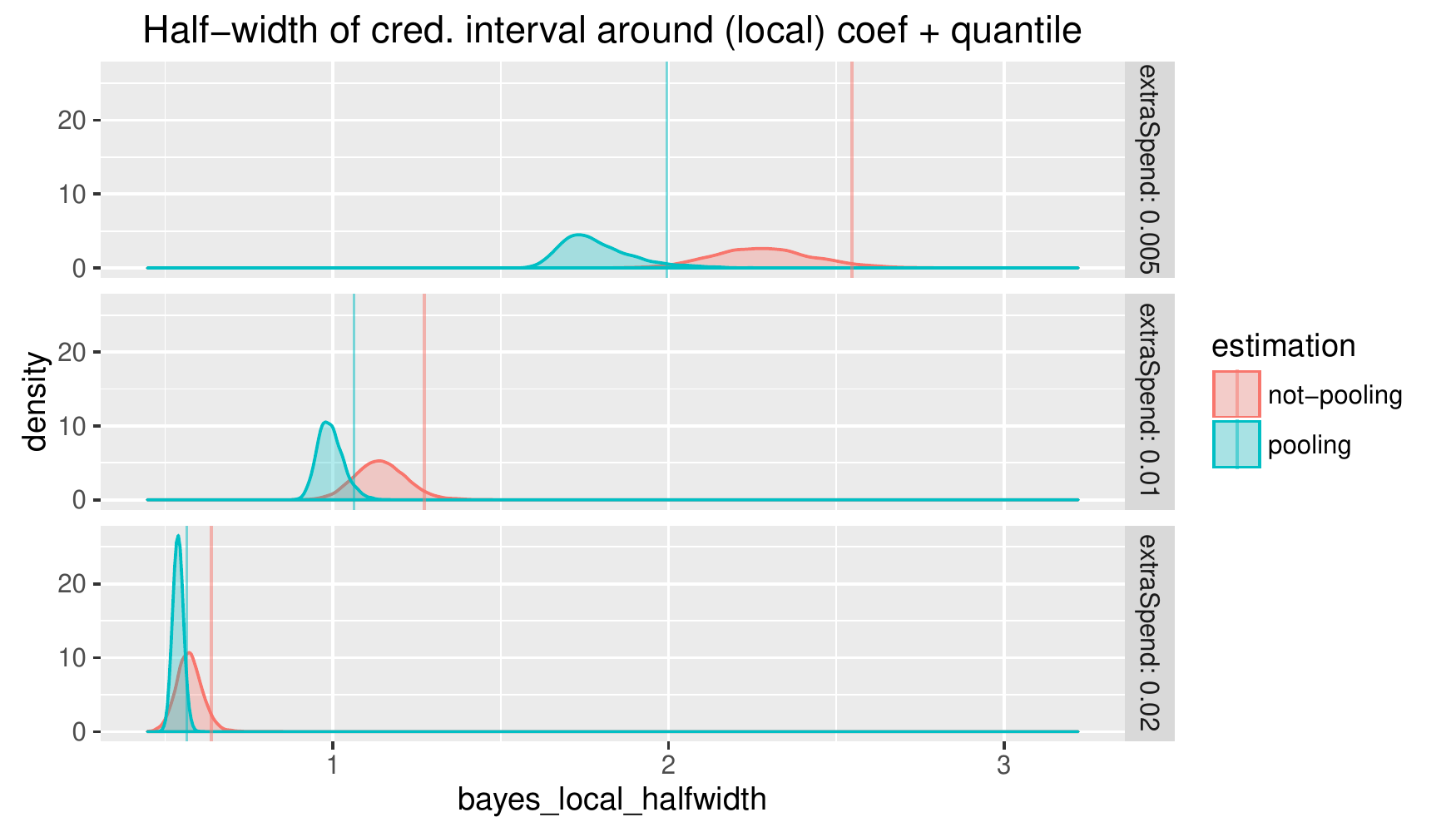}
\caption{Distribution of the half-width of credible intervals around brands' effectiveness coefficients.}
\label{fig:bayes_pool_vs_nopool}
\end{figure}

While the half-width of credible intervals does not quite display an inverse relationship with the extra spend when pooling multiple experiments together as it does when conducting experiments separately (doubling the incremental spend lets us detect twice as small an effectiveness in the single-brand scenario), it is clear that pooling experiments together does bring improvement to the power of the geoexperiments.

It is best if posterior credible intervals have frequentist
coverage levels close to their nominal values.
Table~\ref{tab:coverage} shows empirical coverage levels for $B=4$
brands and $G=160$ GEOS for a range of average returns $\beta$
and brand to brand standard deviations $\sigma_b$.  On the whole
the coverage is quite close to nominal. There is slight over
coverage, probably due to the prior being dominated by large
values of $\sigma_b$.

\begin{table}
\centering
\begin{tabular}{lccccc}  
\toprule
{$\beta \bigm\backslash {\sigma_b}$}  &    0.10 & 0.25  & 0.50  &0.75 &    1.00\\
\midrule
0.25 & 0.974 & 0.974 & 0.967 & 0.953 & 0.950\\
0.50 & 0.975 & 0.967 & 0.962 & 0.954 & 0.946\\
0.75 & 0.974 & 0.968 & 0.967 & 0.959 & 0.952\\
1.00 & 0.978 & 0.975 & 0.969 & 0.964 & 0.954\\
1.25 & 0.975 & 0.968 & 0.966 & 0.960 & 0.958\\
1.50 & 0.978 & 0.968 & 0.966 & 0.954 & 0.954\\
\bottomrule
\end{tabular}
\caption{\label{tab:coverage}
Observed coverage levels of 95\% Bayesian credible
  intervals. Six values of average gain $\beta$ and five
values of brand standard error $\sigma_b$.}
\end{table}

\section{Conclusions and discussion}\label{sec:discussion}

In our examples we see that combining data from multiple
brands at once leads to more accurate experiments than
single brand experiments would yield.  This happens for
both Bayes and empirical Bayes (Stein shrinkage) estimates.
The estimate for any given brand gets better by using data
from the other brands.

This efficiency brings practical benefits. An experiment
on multiple brands might need to use fewer GEOs, or it might
be informative at smaller, less disruptive, changes in the
amount spent.

Ordinary Stein shrinkage towards a common mean 
is advantageous when $B\ge4$ by the theory
of Stein estimation \citep{efro:morr:1973}.
The method of~\cite{xie:kou:brow:2012} is further optimized
to handle unknown and unequal variances.

We have simply plugged in unbiased estimates
of variance.  \cite{hwan:qiu:zhao:2009} propose
a different method that begins with shrinkage
applied to the variance estimates themselves. They also
develop confidence intervals that could be used for
our $\beta_b$.
Stein shrinkage is a form of empirical Bayes estimation.
We found that by using a Bayesian hierarchical
model we could get posterior credible intervals
with good frequentist coverage.

For planning purposes it is worthwhile
to consider what parameter values are realistic
in a specific setting. By simulating several
choices we can find an experiment size that
gets the desired accuracy at acceptable cost.

The most difficult
quantity to choose for a simulation is $\sigma_b^2$,
the variance of the true returns $\beta_b$ to advertising
for different brands.  That is difficult because one
often starts from a position of not having good causal
values for any individual brand.  One more values
for this parameter must then be chosen based on intuition
or opinion.  Because  the true response rate to advertising
can be expected to drift it is reasonable to suppose that
multiple experiments will need to be made in sequence.  Estimates
of $\sigma^2_b$ from one experiment will be useful in planning
the next ones.

\section*{Acknowledgments}

Thanks to Jon Vaver, Jim Koehler, David Chan 
and Qingyuan Zhao for valuable comments.
Art Owen is a professor at Stanford University, but
this work was done for Google Inc., and was
not part of his Stanford responsibilities.

\bibliographystyle{apalike}
\bibliography{multibrand}
\end{document}